\documentclass[journal]{IEEEtran}
\IEEEoverridecommandlockouts
\usepackage{cite}
\usepackage{booktabs}
\usepackage{array}
\usepackage{amsthm}

\newtheorem{lemma}{Lemma}
\newtheorem{theorem}{Theorem}

\newtheorem{assumption}{Assumption}
\usepackage{amsmath,amssymb,amsfonts}
\usepackage{graphicx}
\usepackage{textcomp}
\usepackage{xcolor}
\usepackage{enumerate}
\usepackage{changes}
\usepackage{algorithm}
\usepackage{algpseudocode}
\usepackage{tabularx}
\usepackage{caption,balance}
\usepackage{subcaption}
\usepackage{tikz}
\usepackage{pgfplots}
\usepackage{multirow}
\pgfplotsset{width=2cm,compat=1.8}

\definechangesauthor[color=blue]{me} 

\def\BibTeX{{\rm B\kern-.05em{\sc i\kern-.025em b}\kern-.08em
    T\kern-.1667em\lower.7ex\hbox{E}\kern-.125emX}}

\usepackage{amsthm}
\newtheorem{remark}{Remark}

\begin{document}
\title{Dynamic and Distributed Routing in IoT Networks based on Multi-Objective Q-Learning}
\author{Shubham~Vaishnav, Praveen~Kumar~Donta,~\IEEEmembership{Senior~Member,~IEEE,}~and~Sindri~Magnússon
\thanks{Authors are with the Department
of Computer and Systems Sciences, Stockholm University, 164 55 Kista, Sweden.
 e-mails: \texttt{\{shubham.vaishnav, praveen.donta, sindri.magnusson\}@dsv.su.se}  }
}
\maketitle
\begin{abstract}

IoT networks often face conflicting routing goals such as maximizing packet delivery, minimizing delay, and conserving limited battery energy. These priorities can also change dynamically: for example, an emergency alert requires high reliability, while routine monitoring prioritizes energy efficiency to prolong network lifetime. Existing works, including many deep reinforcement learning approaches, are typically centralized and assume static objectives, making them slow to adapt when preferences shift. We propose a dynamic and fully distributed multi-objective Q-learning routing algorithm that learns multiple per-preference Q-tables in parallel and introduces a novel greedy interpolation policy to act near-optimally for unseen preferences. The algorithm learns to optimize for energy efficiency, packet delivery ratio, and the composite reward, adapting to changing trade-offs between these metrics without retraining or centralized control. A theoretical analysis further shows that the optimal value function is Lipschitz-continuous in the preference parameter, ensuring that proposed greedy interpolation policy yields provably near-optimal behavior. Simulation results show that our approach adapts in real time to shifting priorities and achieves up to 80–90\% lower energy consumption and up to 5 $\times$ higher cumulative rewards and packet delivery compared to six baseline protocols, under dynamic and distributed settings. Sensitivity analysis across varying preference window lengths confirms that the proposed DPQ framework consistently achieves higher composite reward than all baseline methods, demonstrating robustness to changes in operating conditions.

\end{abstract}


\section{Introduction}

The Internet of Things (IoT) has grown rapidly in recent years, connecting billions of devices across domains such as smart healthcare, environmental monitoring, and smart cities. A fundamental task in IoT networks is to sense and transmit data over multi-hop wireless links to destination nodes such as sinks or base stations. Routing in such networks is challenging because nodes are energy-constrained, links are lossy, and traffic requirements vary widely across applications.

Designing efficient routing protocols for IoT requires balancing multiple, often contradictory objectives. Moreover, the relative importance of these objectives is not fixed. For instance, sensor nodes may prioritize energy conservation during routine monitoring, selecting low-power routes that tolerate occasional packet loss. However, during a frost alert, the network may shift to prioritize a higher packet delivery ratio (PDR), possibly using higher-power links, redundant relays, or shorter hops, which can increase overall energy consumption. Future IoT systems must therefore support dynamic preference adaptation, where routing policies adjust to shifting priorities in real time without retraining \cite{liu2017performance, vaishnav2023multi}.

Classical approaches to multiobjective optimization often rely on weighted-sum formulations, where objectives are scaled by fixed weights and combined into a single cost function. While this provides a Pareto front of trade-offs, it assumes that preferences remain static \cite{miettinen1999nonlinear}. In real-world IoT systems, however, priorities may change unpredictably during operation, making static weighting schemes insufficient.

Traditional routing solutions therefore fall short in dynamic IoT settings. Algorithms such as Dijkstra-based methods are computationally expensive for resource-limited devices. Heuristic schemes (e.g., fuzzy logic, clustering, or trust-based protocols) cannot adapt efficiently to changing objectives. Reinforcement learning (RL) approaches \cite{kim2025deep, sharma2024optrisql, wang2024emor, 10419073, 10012574, 10608117, chakraborty2023intelligent,  prabhu2023multiple,  9819969, 9488736, 9685191, 8935210} have emerged as promising alternatives, as they can learn routing decisions from interaction with the environment. However, most existing RL-based works (including deep RL methods) either (i) assume static objectives and retrain policies when preferences change, or (ii) rely on centralized controllers (e.g., SDN-based designs), which face problems of scalability, privacy, and single points of failure \cite{frikha2021reinforcement}. This highlights the need to develop novel RL-based approaches for dynamic large-scale IoT deployments.

\subsection{Motivation}

IoT routing must simultaneously balance multiple performance objectives such as packet delivery ratio (PDR) and energy consumption, while operating on resource-constrained devices. The relative priority of these objectives can differ across applications and may change over time based on network conditions or event-driven demands. A robust routing mechanism must therefore be flexible enough to accommodate varying objective trade-offs and computationally light enough to run in fully distributed IoT environments.

\begin{itemize}
    \item \textbf{Dynamic adaptation to shifting priorities:}  
    IoT applications may require different performance goals at different times, for example, prioritizing PDR during critical events or prioritizing energy conservation during routine monitoring. Routing protocols must adapt to such changing preferences in real time.

    \item \textbf{Generalization across preferences without retraining:}  
    Current RL routing schemes typically learn policies tied to a fixed objective weighting and must be retrained when preferences change. This is slow and energy-intensive for IoT devices. A more efficient approach would allow reuse of past learning to accommodate new preference settings without re-learning policies from scratch.
    
    \item \textbf{Fully distributed and scalable operation for IoT:}  Many existing RL-based routing approaches rely on centralized control or global network information, which is impractical for IoT deployments where nodes operate autonomously with limited resources and only local neighbor knowledge.

\end{itemize}

These considerations motivate the development of a routing protocol that is dynamic, distributed, and able to adapt to changing objective preferences in real time. This work focuses on balancing three key IoT routing objectives: minimizing energy consumption, maximizing PDR, and capturing their context-dependent trade-offs through a composite reward.

\vspace{1em}

\subsection{Contributions and Novelty}

This work introduces a fully distributed, dynamic multi-objective Q-learning routing protocol for IoT networks. The key contributions are:

\begin{enumerate}
    \item \textbf{Parallel per-preference Q-learning under IoT constraints:}  
    A lightweight tabular Q-learning scheme that learns a family of Q-tables for multiple objective preferences in parallel from a single behavior policy, enabling efficient reuse of experience.

    \item \textbf{A novel greedy interpolation policy (GIP) for unseen preferences:}  
    A distributed mechanism that allows each node to construct a routing policy for any unseen preference by interpolating between learned Q-values, enabling fast adaptation to dynamic objective trade-offs without retraining.

    \item \textbf{Theoretical guarantee of near-optimality:}  
    We establish analytical bounds showing that interpolated policies remains
provably near-optimal across the entire range of preferences, with suboptimality controlled by the accuracy of learned Q-values and grid resolution. 
\end{enumerate}

Our experimental evaluation demonstrates that the proposed approach significantly outperforms six existing routing protocols, including both heuristic and reinforcement learning baselines, across diverse preference settings. It achieves up to 80–90\% lower energy consumption and up to 5$\times$ higher cumulative rewards and PDR as compared to the six baseline protocols while maintaining adaptability to dynamic objectives in a fully distributed manner. A comprehensive sensitivity analysis further shows that the proposed distributed preference-aware Q-learning approach remains robust across different temporal preference granularities, consistently outperforming existing routing schemes in terms of the composite reward.

\color{black}

The paper is organized as follows: Section~\ref{sec:relatedwork}  discusses related work, Section~\ref{section: ModelAndFormulation} outlines the system model and problem formulation, Section~\ref{sec:Algorithm} details the proposed algorithm, Section~\ref{sec:DistAlg} illustrates how the algorithm can be implemented in fully distributed manner, Section~\ref{sec:Simulations}  presents our simulation results, and Section~\ref{sec:conclusion} concludes.


\section{Related Works}\label{sec:relatedwork}

\newcolumntype{L}[1]{>{\raggedright\arraybackslash}p{#1}}
\newcolumntype{C}[1]{>{\centering\arraybackslash}p{#1}}

\begin{table*}[t]
\centering
\scriptsize
\caption{Comparative summary of 
RL-based routing algorithms.}
\label{tab:rl-routing-summary}
\setlength{\tabcolsep}{3pt}
\renewcommand{\arraystretch}{1.1}
\begin{tabular}{
C{0.09\textwidth}  
L{0.22\textwidth}  
L{0.22\textwidth}  
C{0.08\textwidth}  
C{0.09\textwidth}  
C{0.10\textwidth}  
C{0.08\textwidth}  
}
\toprule
\textbf{Ref} & \textbf{Reinforcement Learning (RL) Method} & \textbf{Multiobjective Optimization (MO) Method} &
\textbf{Adaptation to dynamic preferences} & \textbf{MO Pareto Approximation}  & \textbf{Theoretical guarantees for MO method} & \textbf{Fully Distributed} \\
\midrule

\cite{kim2025deep} & DRL + Intelligent neighbor node selection         & - & -      & -   & -  & \checkmark \\

\cite{sharma2024optrisql, chakraborty2023intelligent, farag2021congestion }  & Tabular QL            & Scalarised linear reward                       & -   & -  & -   & \checkmark  \\ 

\cite{wang2024emor}  & RL with Actor-Critic architecture & RL-learned weights (energy vs.\ progress) & Likely  & -  & -  & -  \\

\cite{10419073} & A3C Based Multistrategy RL & Ratio-based reward & -      & -   & -  & - \\

\cite{10012574} &  Graph Neural Network (GNN) structure
in Deep-RL & - & -      & -   & -  & - \\ 

\cite{prabhu2023multiple} & Multi-policy Deep-RL & Ratio-based rewards & -     & -        & -  & - \\

\cite{8935210} & Model free DRL (Actor-Critic) & Linear scalarized reward & -      & -   & -  & - \\

\cite{kaur2021energy} & Deep-RL (DRL)       &  Linear scalarization (LS) function   & -      & -  & \checkmark & - \\

\hline

\textbf{Our work} & \textbf{Dynamic Preference Q-Learning with Greedy Interpolation} &  \textbf{\checkmark} &  \textbf{\checkmark}  &  \textbf{\checkmark}  &  \textbf{\checkmark} &  \textbf{\checkmark} \\

\bottomrule
\end{tabular}
\end{table*}

Over the last two decades, there has been an upsurge in research related to developing efficient routing protocols for wireless networks and IoT. Some works present enhancements in Dijkstra’s algorithm for routing \cite{5158843, 7435250}. However, such routing protocols are often associated with high local computations, storage requirements, and time complexities. Thus, it is inefficient to utilize them for resource-contrained IoT networks today since they have to deliver quality services to users. Also, some other heuristic-based methods like fuzzy logic \cite{5158843}, trust-based or cluster-based approaches such as MANET and VANET are also unable to meet the multiobjective expectations from IoT of the present and future. With the evolution of Machine Learning (ML), many ML-based approaches have been presented as promising solutions to routing-related optimization problems. Some of these routing algorithms are built on ML paradigms such as clustering, and evolutionary algorithms \cite{10516314}.

RL, among other data-driven ML approaches, has attracted special attention since it provides us with the facility to learn optimal policies while interacting with the environment. Thus, many RL-based routing protocols have been proposed \cite{sharma2024optrisql, chakraborty2023intelligent, farag2021congestion, wang2024emor, prabhu2023multiple, kaur2021energy, kim2025deep, 10419073, 10608117, 8935210, abadi2022rlbeep, bouzid2020efficient, guo2019optimizing, 9819969, 9634122, 9488736, cong2021deep, ergun2022reinforcement, 9685191, 10012574, phan2022deepair}. A direct utilization of RL for IoT is to rely upon a central controller for route planning and control. This has been made possible by the recent development of software-defined networking (SDN). In SDN-enabled IoT networks, a central controller has complete knowledge of the network topology and takes most of the control decisions in the network.
The SDN-based centralized protocols \cite{8935210, 10419073, 9634122, phan2022deepair, 9488736, cong2021deep, 10012574, 9819969} face several challenges. Since the SDN controller plays a pivotal role in such protocols, it is an attractive target for malicious users to carry out DDoS attacks or other security attacks \cite{javanmardi2023sdn}. Such routing protocols are also vulnerable to single-point failures. Another critical issue in IoT is maintaining individual IoT nodes' data privacy. Sharing all the data with a central controller can be both cost-inefficient and also pose a risk to data privacy.
Moreover, as the size of IoT networks is increasing rapidly, scalability is a major issue with these centralized protocols \cite{cong2021deep}.

With a more intelligent organization, it is feasible and beneficial to employ distributed RL algorithms for routing. 
As an example, in tabular Q-learning, SDN-based and other centralized approaches rely on a centralized controller to keep the knowledge of the entire network, store the entire Q-table, and do route planning as the RL agent. However, the Q-table can be split into smaller parts, with each node having its local small Q-tables with information about and interaction with their immediate neighbors. In this way, each node acts as an RL agent capable of making its own decisions. Some existing research works utilize this possibility and thus propose distributed RL-based routing protocols \cite{guo2019optimizing, bouzid2020efficient}. Abadi et al. \cite{abadi2022rlbeep} utilize a sleep scheduling mechanism for enhancing the energy efficiency of their RL-based routing protocol. Such sleep scheduling is difficult to achieve in a fully distributed setup. Thus, they utilize a semi-distributed approach. They divide the network into clusters, with each cluster head acting as the local controller for the cluster. In this work, however, we showcase how a fully distributed RL-based scheme can be effectively applied while ensuring high energy efficiency and other vital objectives. Some other distributed RL-based routing protocols \cite{sharma2024optrisql, chakraborty2023intelligent, farag2021congestion, kim2025deep} have also been proposed, but they lack in dynamically adapting to changing preferences, as described below (Comparisons are summarized in Table~\ref{tab:rl-routing-summary}).

The Routing Protocol for Low-Power and Lossy Networks (RPL) is a protocol standardized by the Internet Engineering Task Force (IETF) for Low-Power and Lossy Networks, and it is widely used in many Internet of Things (IoT) deployments. In RPL, routing decisions are primarily driven by the objective function, which determines parent selection based on one or more routing metrics. As a result, a growing body of work has proposed reinforcement learning–enhanced RPL variants that aim to improve the adaptability of this decision process. Most RL-based RPL schemes formulate routing as a single-objective problem by combining multiple metrics—such as packet delivery ratio (PDR), delay, and energy—into a scalar reward using fixed or heuristically chosen weights \cite{farag2021congestion, homaei2021ddsla, homaei2025load}. Even when weights are adjusted online, for example through learning-automata-based tuning (e.g., DDSLA-/LALA-RPL \cite{homaei2021ddsla, homaei2025load}) or context-aware objective functions \cite{farag2021congestion}, the learned policy ultimately converges to a single operating point for a given application profile. Consequently, adapting to changing application requirements typically requires re-tuning reward weights or re-running the learning process. 

Our proposed scheme overcomes this limitation by supporting dynamic preference changes without retraining. While no existing work demonstrates a fully dynamic, distributed, multi-objective Q-learning routing scheme implemented strictly within RPL’s standard Objective Function and rank computation on resource-constrained nodes, several studies provide strong supporting evidence for the feasibility of such an approach. Prior work has shown that Q-learning can be effectively used for parent selection in RPL-based routing \cite{farag2021congestion, lalani2024quera, sebastian2018load}, and other studies demonstrate that learning-based decision logic can be embedded within RPL objective functions without modifying the protocol’s control plane \cite{kuwelkar2023rpl}. Building on these insights, our proposed dynamic multi-objective Q-learning scheme can be mapped to RPL’s Objective Function and rank computation, enabling learned parent selection while preserving RPL’s existing control messages and DODAG structure. Moreover, prior implementations and evaluations in constrained IoT environments, such as those using Contiki and the Cooja simulator, suggest that such learning-based extensions to RPL are practically feasible \cite{sebastian2018load}.

More generally, existing RL-based routing approaches struggle to adapt quickly to changes in preference settings, as they require time to relearn policies suited to new objectives. For example, Natarajan et al. \cite{natarajan2005dynamic} introduce a distributed multi-objective reinforcement learning (MORL) approach with dynamic preferences. However, this method does not approximate the full Pareto front and instead learns from a limited number of random policies encountered during runtime. Furthermore, it is restricted to average-reward reinforcement learning algorithms such as R-learning and H-learning and does not extend to discounted Q-learning. These limitations are important in IoT settings, where learning efficiency and computational simplicity are critical. Prior work has shown that Q-learning consistently outperforms R-learning, even when both are evaluated using the same undiscounted performance metric \cite{mahadevan1994discount}. In addition, R-learning is highly sensitive to the choice of exploration strategy, with performance degrading significantly under Boltzmann exploration, while H-learning incurs substantially higher computational overhead and slower convergence compared to model-free methods such as Q-learning \cite{tadepalli1994h}. These properties make average-reward methods less suitable for routing in resource-constrained IoT networks.

To overcome the above-mentioned challenges, we propose a novel multi-objective Q-learning–based distributed routing scheme that can exploit the learned Pareto frontier to adapt efficiently and swiftly to dynamically varying preferences in IoT networks. Our approach learns multiple Q-functions corresponding to different preference vectors and introduces a greedy interpolation policy that enables near-optimal routing decisions for previously unseen preferences without restarting learning or relying on centralized coordination. Importantly, we provide formal performance guarantees by proving Lipschitz continuity of the optimal value function with respect to preferences, which yields explicit bounds on the sub-optimality of interpolated policies. To the best of our knowledge, existing RL-based and other routing proposals do not support preference-continuous policies or offer theoretical guarantees under dynamic preference shifts.

\section{System Model and Problem Formulation}
\label{section: ModelAndFormulation}
 \subsection{Network Model and Routing Algorithms} \label{sec:MandA}


Consider an IoT-enabled network of $N$ nodes denoted by the set $\mathcal{N}=\{1,\ldots, N\}$. The nodes can communicate over a  directed graph $\mathcal{G}=(\mathcal{N},\mathcal{E})$, where  $\mathcal{E}\subseteq \mathcal{N}\times  \mathcal{N}$ is a set of edges. In particular, a node $i\in \mathcal{N}$ can communicate data to node $i'\in \mathcal{N}$ if and only if there is an edge from $i$ to $i'$, i.e., if $(i,i')\in \mathcal{E}$. 
In the network, data transmission between nodes may experience packet loss. For each edge \( (i, j) \) in the graph, let \( P(i,j)\in (0,1] \) represent the probability of packet loss when transmitting from node \( i \) to node \( j \). It is important to note that the  probabilities \( P(i,j) \) do not need to be known.

   We focus on the problem of network routing, which involves determining an optimal path for transmitting a packet from a source node to a destination node within the network.
 Network routing is naturally formulated as a Stochastic Shortest Path Markov Decision Process (MDP)~\cite{bertsekas2012dynamic}.
 An MDP models sequential decision-making problems where an agent interacts with its environment over a series of steps. Formally, an MDP consists of a set of states $\mathcal{S}$, a set of actions $\mathcal{A}$, a transition function $T:\mathcal{S}\times \mathcal{A}\rightarrow \Delta( \mathcal{S})$\footnote{\( \Delta(\mathcal{S}) \) denotes the set of all probability distributions over the state space \( \mathcal{S} \). The transition function \( T(s, a) \) specifies a probability distribution over the next states given the current state \( s \) and action \( a \).}, and a reward function $r:\mathcal{S}\times \mathcal{A}\rightarrow \mathbb{R}$. In a Stochastic Shortest Path MDPs we have an absorbing terminal $s_T\in \mathcal{S}$, meaning that once it is reached, no further transitions occur, effectively ending the decision-making process.
 
 The agent interacts with its environment over sequence of episodes. In  each episode, the agent starts in some initial state $S^0\in \mathcal{S}$ and at each step selects an action based on the current state. The environment then transitions to a new state according to the transition function, and the agent receives a reward as determined by the reward function. This process generates a trajectory 
 \begin{equation} \label{eq:Trajectory}
    S^0, A^0, R^0, \dots, S^K, A^K, R^K,S^{K+1},
 \end{equation}
 where $S^{K+1}=s_T$ and the sequence of states evolves according to $S^{k+1} \sim T(S^k,A^k)$. The rewards are given by $R^k = r(S^k, A^k)$.
 At each time step, the agent decides on the next action by following some policy, a mapping $\pi:\mathcal{S}\rightarrow \Delta(\mathcal{A})$ from state to a distribution of actions. 
  The value of the policy $\pi$ is defined as the expected cumulative rewards, mathematically defined as
$$ V_{\pi}(s) =\mathbf{E}\left[ \sum_{k=0}^{K} R^k| S_0=s \right].$$ 
  The agent's goal is to find an optimal policy $\pi^{\star}$ that maximizes the value function.

%
 %
 %
%
%

 In the context of routing, the state space is given by:
$$\mathcal{S}=\mathcal{N}\times\mathcal{N}\cup \{s_T\}.$$
Here, each state $(i,j)\in \mathcal{N}\times\mathcal{N}$ indicates that the packet is currently at node $i$ with the destination being node $j$.  The actions space $\mathcal{A}$ corresponds to the possible decisions or routing choices available at each state. In particular, for state $s=(i,j)\in \mathcal{N}\times\mathcal{N} $  the action space is
$$\mathcal{A}(s)=\{i'\in \mathcal{N}| (i,i')\in \mathcal{E}\}$$
where $a\in \mathcal{A}(s)$ represents selecting the next node $a\in \mathcal{N}$ to which the packet should be forwarded from node $i$ on the path to the destination $j$. The terminal state $s_T$ has no actions. The transition function for a state $s=(i,j)$ where $i\neq j$ is
$$T(s,a)=\begin{cases} (a,j) & \text{ with probability } 1-P(i,a) \\ s_T &\text{ with probability } P(i,a)  \end{cases}$$
 and if $i=j$ it is simply $T(s,a)=s_T$.


A routing policy \( \pi: \mathcal{S} \rightarrow \Delta(\mathcal{A}) \) dictates the action to be taken in a given state \( s = (i, j) \), when the  current location of the packet is node $i$ and  the destination node is \( j \). Specifically, \( \pi(s) \) determines the probability of the next node to which the packet should be forwarded from node \( i \), guiding its progression toward the destination node \( j \). The resulting trajectory is as given in Eq.~\eqref{eq:Trajectory}, where $S^k=(I^k,J^k)$ for $k=0,1,\ldots,K$. Here $I^0$ is source node and the sequence of nodes $I^0,I^1,\dots,I^{K}$ generates a path from source to the destination $J^0$ or to the last node receiving the package before the package drops. 

A key objective of routing is to find the routing policy that optimizes specific metrics, such as energy consumption or Packet Delivery Ratio (PDR). In this context,
 each edge $(i,j)\in \mathcal{E}$  can be associated with a quantity $E(i,j)$ indicating, e.g., the energy needed to transmit data from  $(i,j)\in \mathcal{E}$. 
  We can now define the  energy reward for non-terminal states $s=(i,j)$ and action $a$ as\footnote{The reward is the negative of the true energy, as the objective is maximization.}
 $$r^{\texttt{Energy}}(s,a)= \begin{cases}
     -E(i,a) & \text{if} ~~i\neq j \\
     0      & \text{otherwise.}
 \end{cases}$$
 The energy consumption of an episode in which the agent follows policy $\pi$, starting from state $S^0=(i,j)$, is then given by the random variable
  $$\texttt{Energy}(i,j)=\sum_{k=0}^{K} r^{\texttt{Energy}}(S^k,A^k).$$
  We can now define an energy value function, which is simply the expected value
  $$ V_{\pi}^{\text{Energy}}(i,j)= \mathbf{E} \left[ \texttt{Energy}(i,j) \right]. $$
This value function quantifies the average energy consumption, when transmitting a packet from source $i$ to destination $j$.

There can be  other metrics that we would like to optimize in addition to Energy. Another common metric is PDR.
 To characterize the PDR for a given policy $\pi$ and a pair of source $i$ and destination $j$, consider the path from $i$ to $j$ when following $\pi$, i.e.,  $i^0,\ldots,i^K$ where $i^0=i$, $i^K=j$ and $i^{k+1}=\pi(i^k,j)$ for $k=0,1,2,\ldots,K-1$.
 Then the \texttt{PDR} for the policy $\pi$ is simply 
  \begin{equation}\label{eq:value_PDR}
      V_{\pi}^{\texttt{PDR}}(i,j)=\Pi_{k=1}^K \left(1-P(i_k,i_{k+1})\right).
  \end{equation}
  However, the PDR cannot be computed since we do not know the probabilities $P(i,j)$. However, if we define the reward for state $s=(i,j)$ and action $a$ as
\begin{equation}
     r^{\texttt{PDR}}(s,a) = 
     \begin{cases}
        1,  &  i=j
        \\
        0, & \text{Otherwise}.
        \end{cases}
\end{equation}
This reward is easily measurable, since the receiving node will know if the packet is received or not. Moreover, it is easily established that
$$ V_{\pi}^{\texttt{PDR}}(s)= \mathbf{E} \left[ \sum_{k=0}^{K} r^{\texttt{PDR}}(S^k,A^k) | S^0=s\right],$$
that is $V_{\pi}^{\texttt{PDR}}(s)$ is the value function for the reward $r^{\texttt{PDR}}(s,a)$, and this is why we defined the PDR in Eq.~\eqref{eq:value_PDR} using the notation for a value function.

\subsection{Routing Dynamic Preference}\label{subsec:MOOP}

While most research on routing has traditionally focused on optimizing static objectives such as energy consumption, packet delivery, or latency, real-world scenarios demand more dynamic approaches. In real-world systems, the priority of these objectives can shift rapidly. For instance, when the system has ample energy reserves, optimizing for a higher PDR might be the primary goal to ensure reliable communication. However, as energy levels deplete, conserving power becomes more critical, shifting the focus towards energy efficiency. In other cases, a balanced trade-off between energy usage and PDR might be required, necessitating adaptive strategies that can scale priorities based on current conditions.

 Our goal is to consider Dynamic Preference Routing, where at each episode the preference between the different objective can change. To simplify our presentation, we consider two primary objectives, energy $V_{\pi}^{\rm{Energy}}(s)$ and $V_{\pi}^{\texttt{PDR}}(s)$. However, these could easily be replaced or expanded to include other objectives of interest.
 
 We focus on a setup, where routing decisions are made  over a sequence of $M$ episodes, where $m=1,2,3,\ldots,M$ is the episode index. At the start of each episode $m$, the agent receives a preference indicator  $\beta_m\in [0,1]$, which determines the relative importance of each objective for that specific episode. Specifically, the larger the value of \( \beta_m \), the greater the preference given to the energy objective and the lesser the emphasis on the PDR objective, and vice versa. The agent's task is to find the optimal routing policy, $\pi$, that maximizes the combined, scaled accumulated reward:
\begin{align*} V_{\pi}(s)
= \mathbf{E} \Big[ \sum_{k=0}^{K}  \beta_m r^{\texttt{Energy}}&(S^k,A^k) \\ &+(1-\beta_m)r^{\texttt{PDR}}(S^k,A^k) | S^0=s\Big].
\end{align*}
 It is easily established that
 $$V_{\pi}(s) =\beta_m V_{\pi}^{\rm{Energy}}(s)+(1-\beta_m) V_{\pi}^{\texttt{PDR}}(s).$$
%
 The goal is thus, at each episode $m$, to find the policy $\pi_m$ that optimzies
\begin{equation} \label{Eq:MainFormulation}
\begin{aligned}
& \underset{\pi_m}{\text{maximize}}
& & \beta_m V_{\pi_m}^{\rm{Energy}}(s)+(1-\beta_m) V_{\pi_m}^{\texttt{PDR}}(s).
\end{aligned}
\end{equation}
This means that the agent must dynamically adapt its routing strategy, continuously optimizing for different objectives as the preference vector shifts from one episode to the next. This is particularly challenging because it requires the agent to make real-time adjustments in an unpredictable environment, where the trade-offs between objectives are constantly changing. Our central contribution is to develop a framework that enables the agent to effectively learn and implement these adaptive strategies, ensuring optimal performance across varying and evolving conditions.

\section{Learning Based Routing with Dynamic Preferences} \label{sec:Algorithm}

 In this section, we illustrate our learning based algorithm for  routing with dynamic preferences. We draw on two main ideas, Q-learning and multi-objective optimzation. We first give the background on Q-learning in subsection~\ref{Sec:Q-learning}, followed by a detailed description of our algorithm in subsection~\ref{Sec:Q-DP-learning}. Finally, we explore various implementation strategies for exploration in subsection~\ref{Sec:AlgorithmExploration}.

\subsection{Q-Learning} \label{Sec:Q-learning}

To learn the optimal policy in MDPs from experience or sequences of trajectories, reinforcement learning (RL) algorithms are particularly effective. 
To find the optimal action in any given state, it is useful to introduce the state-action value function, commonly known as the Q-table. The Q-table is essentially a function that maps state-action pairs to expected accumulated rewards if we follow the optimal policy. More formally, the Q-value for a state-action pair \((s, a)\) is defined as the expected value of the accumulated rewards, conditioned on taking action \(A\) in state \(S\) and then following the optimal policy thereafter. Mathematically, this can be expressed as:
\[
Q(s, a) = \mathbb{E}\left[ \sum_{k=0}^{K} R^{k} \mid S_0 = s, A_0 = a \right].
\]
 If we know the Q-table then the optimal policy can be easily recovered by selecting the action that maximizes the Q-value for each state:
 $$\pi(s)=\underset{a\in \mathcal{A}}{\text{argmax}} ~Q(s,a).$$
 However, to determine this optimal policy, we first need to learn the Q-table.

Q-learning is an algorithm designed to learn the optimal Q-values through interaction with the environment. At each step, the algorithm receives a sample consisting of the current state \(S\), the action taken \(A\), the next state \(S'\), and the reward \(R\) obtained after transitioning to \(S'\). The Q-table is then updated using the following rule:
\[
Q(S,A) = Q(S,A) + \alpha \left( R +  \max_{A'} Q(S',A') - Q(S,A) \right).
\]
Here, the update adjusts the Q-value by incorporating both the immediate reward and the maximum estimated future reward from the next state. 
The parameter \(\alpha\) is the learning rate, it plays a critical role in determining how quickly or accurately the Q-learning algorithm converges to the optimal Q-values. When \(\alpha\) is constant, the Q-values update by a fixed proportion after each state-action pair is explored, leading to an approximate solution. As \(\alpha \to 0\), the updates become smaller, allowing the algorithm to converge more closely to the true optimal Q-values. However, this comes at the cost of a slower convergence rate, as smaller updates lead to more gradual refinement of the Q-table. 

To address this trade-off, \(\alpha\) can also be set to decrease gradually each time a state-action pair is explored. By defining \(\alpha\) as a function of the number of times a particular pair \((s, a)\) is visited—denoted as \(\alpha(s, a)\)—we can ensure more precise updates as learning progresses. Specifically, if the learning rate is chosen such that it is summable but not square summable, i.e., \(\sum_t \alpha_t(s,a) = \infty\) and \(\sum_t \alpha_t^2(s,a) < \infty\), it is guaranteed that the algorithm will converge to the true optimal solution. 



\begin{remark}\label{remark:off-policy} A key advantage of Q-learning, which we leverage in our approach, is that it is an off-policy algorithm~\cite{sutton1999reinforcement}. This means that the learning process is independent of the policy being followed during sample collection. In other words, Q-learning can learn the optimal Q-values regardless of the strategy used to explore the environment. As long as every state-action pair is explored sufficiently often, Q-learning is guaranteed to converge to the optimal Q-table. This property is particularly useful in multi-objective scenarios, since  we can learn optimal Q-tables for many preferences simultaneously, irrespective of the actual routing policy being followed during the learning phase.
\end{remark}

\subsection{Dynamic Preference Q-Learning} \label{Sec:Q-DP-learning}
Direct application of Q-learning is not well-suited for addressing dynamic preferences in routing, as it struggles with changing rewards and multiple objectives. In this section, we outline our main algorithmic concepts to tackle this problem.

\begin{algorithm}[t]
\begin{algorithmic}[1]
\caption{\texttt{DPQ-Learning}}\label{alg:DPQ}
\State \texttt{Initialize the grid} $\mathcal{B}$ \texttt{and the Q-tables:}
\For {$\beta\in \mathcal{B}$}
    \State \texttt{Initialize} $Q_{\beta}~\in \mathbb{R}^{|\mathcal{S}|\times|\mathcal{A}|}$  ~~~~ \# E.g put to zero 
\EndFor
\For {each episode $m=1,2,\ldots,M$}
    \State \texttt{Receive} $S^0=(i,j)$.~~~~~~ \# This is the routing task
    \State \texttt{Receive preference vector} $\beta_m$
    \State  \texttt{Initialize iteration count} $k=0$
    \While{$S^k\neq s_T$} ~~~~\# Terminal state is not reached
        \State \texttt{Take action} $A^k\sim \pi_m(S^k)$
        \State \texttt{Receive next state}  $S^{k+1}$ 
        \For {each $\beta$ in $\mathcal{B}$}
            \State $R_{\beta}^k \gets (1-\beta) r^{\text{Energy}}(S^k,A^k)+ \beta r^{\texttt{PDR}}(S^k,A^k)$
                
            \State $Q_{\beta}(S^k, A^k) \gets Q_{\beta}(S^k, A^k)  +  \alpha \Big(R_{\beta}^k + $  \State  \hspace{1.2cm} $\max_{a\in \mathcal{A}} Q_{\beta}(S^{k+1}, a) - Q_{\beta}(S^k, A^k) \Big)$    
        \EndFor
    \EndWhile   
\EndFor
\end{algorithmic}
\end{algorithm}

When we have dynamic preferences, the objective can change quickly, and we must be able to respond or learn to respond rapidly to any preference $\beta$. To achieve this, our algorithm employs a Q-table for each preference parameter $\beta$, leveraging the off-policy nature of Q-learning to learn all Q-functions in parallel, see Remark~\ref{remark:off-policy}. This means that, even while using a policy optimized for one specific $\beta$, we can still use the collected data to update the Q-tables for other parameters. This approach ensures that we remain adaptable and can handle shifting objectives efficiently. However, directly applying this approach is computationally infeasible. 
Thus, we consider a finite grid or a subset \( \mathcal{B} \subseteq [0,1] \) of preferences that we learn, and then interpolate for \( \beta \) that is not in \( \mathcal{B} \). Specifically, for each \( \beta \in \mathcal{B} \), we maintain a Q-table \( Q_{\beta} \), and these \( Q_{\beta} \) are learned in parallel during the routing process. 

We illustrate this processes in Algorithm~\ref{alg:DPQ}. We start by initializing the algorithm with the grid $\mathcal{B}$ and by initializing the Q-table $Q_{\beta}$ for each $\beta\in\mathcal{B}$. A straightforward initialization is to set all the Q-tables to zero, i.e., set $Q_{\beta}(s,a)=0$ for all $(s,a)\in \mathcal{S}\times \mathcal{A}$. After that we run a sequence of $M$ episodes, indexed by $m=1,2,\ldots,M$. Each episode begins with receiving an initial state $S^0=(i,j)$. The initial state indicates the next routing task where $i$ is the source and $j$ is the destination. We also receive a preference $\beta_m$ indicating how to balance the objectives during the routing task of the episode, as described in Eq.~\eqref{Eq:MainFormulation}.


In each episode $m$, we follow the routing policy $\pi_m(\cdot)$, which generates a trajectory of states, actions, and rewards. This trajectory terminates when the system reaches the terminal state $s_T$, either because the packet is dropped or it successfully reaches its destination. From each sample $(S^k, A^k, S^{k+1})$ in the trajectory, we update all the Q-tables $Q_{\beta}$ for each $\beta \in \mathcal{B}$. It is guaranteed that these Q-tables $Q_{\beta}$ will converge to the true optimal Q-values, provided the policies include sufficient exploration, i.e., there is a positive probability of selecting every possible action, see discussion in previous subsection. 


A key part of each episode $m$, is to select a new routing policy $\pi_m(\cdot)$. This policy is based on the preference $\beta_m$ for that episode, and ideally, $\pi_m(\cdot)$ should be optimal with respect to the given preference, aiming to solve the optimization problem in Eq.~\eqref{Eq:MainFormulation}. However, learning the optimal policy requires trial and error. In the following section, we discuss how to efficiently learn the optimal policy for dynamic preferences.

\subsection{Exploration Strategies} \label{Sec:AlgorithmExploration}

When learning from experience using RL, a key challenge is balancing the trade-off between exploration and exploitation. Exploration involves trying out new routes to discover potentially more efficient options, while exploitation leverages the knowledge already gained to select the current best-known routes. Achieving this balance is crucial for ensuring the algorithm does not get stuck in suboptimal routing decisions while still converging towards optimal performance over time.

Pure exploration involves selecting a random action in each state. Formally, we define the pure exploration policy as $\pi^{\texttt{PE}}: \mathcal{S} \rightarrow \Delta(\mathcal{A})$, where $\pi^{\texttt{PE}}(s) = a $ with probability $1/|\mathcal{A}(s)|$
for each $a \in \mathcal{A}(s)$.
Pure exploitation, in contrast, involves selecting the best action based on the current estimate of the optimal $Q$-table. In RL, this is commonly referred to as the greedy policy. In our setting, at episode $m$, the greedy policy selects the action that maximizes the current $Q$-value with respect to the preference vector $\beta_m$. If $\beta_m$ is not in the set $\mathcal{B}$, we determine the greedy policy by interpolating between the nearest parameters in $\mathcal{B}$. 
We illustrate this in Algorithm~\ref{alg:greedy}.
\begin{algorithm}
\begin{algorithmic}[1]
\caption{Greedy Interpolation Policy (GIP)}\label{alg:greedy}
\State \texttt{Receive} $\beta$
\If {$\beta \in \mathcal{B}$}
    \State $Q_{\beta}^{\texttt{Int}}\gets Q_{\beta}$
\Else
    \State $\overline{\beta}\gets \min \{ b \in \mathcal{B} | b\geq \beta \}$
    \State $\underline{\beta} \gets \max \{ b \in \mathcal{B} | b\leq \beta \}$
    \State $\rho \gets (b-\underline{\beta})/(\overline{\beta}-\underline{\beta})$
    \State $Q_{\beta}^{\texttt{Int}}\gets \rho Q_{\underline{\beta}}+ (1-\rho)Q_{\overline{\beta}}  $
\EndIf
\State $\pi_{\beta}^{\texttt{Greedy}}(s):=\underset{a\in \mathcal{A}}{\text{argmax}}~ Q_{\beta}^{\texttt{Int}}(s,a)$
\end{algorithmic}
\end{algorithm}

There are several approaches to efficiently balance the trade-off between exploration and exploitation. These methods typically focus on greater exploration in early stages when little is known about the system. As learning progresses and the Q-tables approach optimal values, more emphasis is placed on exploitation. In this paper, we examine two exploration strategies, though others can be integrated with our algorithm. The strategies we study are: (a) sequential exploration and exploitation, and (b) simultaneous exploration and exploitation.

\textbf{Sequential Exploration and Exploitation}
is a straightforward approach where the first $m_{\texttt{exp}}<m$ episodes are dedicated entirely to exploration. During these initial $m_{\texttt{exp}}$ episodes we employ   the random policy $\pi^{\texttt{PE}}(\cdot)$. The remaining episodes focus solely on exploitation where we employ the greedy policy $\pi_{\beta}^{\texttt{Greedy}}(\cdot)$ in Algorithm~\ref{alg:greedy}. This method ensures that the agent thoroughly explores its environment early on, gaining a broad understanding of possible actions and outcomes. However, a key drawback is that it lacks flexibility, if too many episodes are devoted to exploration, it may not exploit effectively in later stages, while too few exploratory episodes can lead to a suboptimal policy. Additionally, this rigid separation does not adapt to the learning progress, which could hinder efficiency in dynamic environments.

\textbf{Simultaneous Exploration and Exploitation} is an adaptive approach where exploration and exploitation occur together throughout the learning process. Early on, an epsilon-greedy strategy is used, meaning the agent selects random actions with a high probability (epsilon) to encourage exploration, while also exploiting the best-known actions at a smaller rate. Over time, epsilon is gradually reduced, placing less emphasis on exploration and more on exploitation as the agent learns. This method provides a more flexible balance between exploration and exploitation, allowing the agent to adapt as it gains knowledge. However, it can be slower to converge compared to purely exploiting after a set exploration phase.

\subsection{Objectives and Routing Metrics}

The proposed approach is tailored to meet real-world IoT routing needs and evaluates routing decisions based on three key objectives:

\begin{itemize}
    \item \textbf{Energy Consumption:} Each transmission consumes energy depending on the hop distance. Minimizing cumulative energy usage is critical for IoT nodes with limited battery life.
    
    \item \textbf{Packet Delivery Ratio (PDR):} This represents the reliability of a path and is influenced by link quality, node availability, and hop count. Maximizing PDR is essential in scenarios with critical data delivery needs.
    
    \item \textbf{Composite Reward:} A weighted combination of the above metrics based on user-defined preference parameter $\beta \in [0,1]$, allowing the system to dynamically prioritize energy efficiency or reliability during runtime.
\end{itemize}

While the learning framework is general, the algorithm is explicitly implemented for multi-hop IoT routing under realistic constraints: distributed operation, lossy links, and battery-limited nodes. The preference-driven reward design and distributed learning scheme are tailored to unique challenges in IoT environments, distinguishing this work from generic routing strategies.

\section{Distributed DPQ-Learning } \label{sec:DistAlg}

 While the DPQ-Learning algorithm in the previous section offers a centralized approach to dynamic routing, a distributed solution is often more desirable in practice, especially for large-scale or decentralized networks. Distributed algorithms enhance scalability and privacy, reduce communication overhead, and improve robustness by avoiding a single point of failure. We now illustrate how the DPQ-Learning algorithm can be implemented in a fully distributed manner, where each node only needs to exchange information with its immediate neighbors in the routing network.


In our distributed setup, each node \(i\) maintains a local policy and a corresponding set of local Q-tables. Specifically, each node \(i\) retains a Q-table \(Q_{\beta}^i\) for each \(\beta \in \mathcal{B}\). The local Q-tables \(Q_{\beta}^i\) for all nodes $i$ form a disjoint subset of the global Q-table \(Q_{\beta}\) for the preference $\beta$. In particular, for each state \(s = (i,j) \in \mathcal{S}\) and action \(a \in \mathcal{A}(s)\), the local Q-table satisfies the condition 
$$Q_{\beta}^i(j,a) = Q_{\beta}(s,a).$$
%
This implies that each node \(i\) maintains a local Q-table, which allows the node to independently update its estimates based on local information and interactions.

If \(Q^i_{\beta}\) is the true optimal Q-table for each node $i$ then the optimal policy at state \(s = (i,j)\) can be found by selecting the action \(a\) that maximizes the Q-value, i.e., 
\[
\pi^i_{\beta}(j) = \underset{a }{\arg \max}\, Q^i_{\beta}(j, a).
\]
This is a fully local policy, meaning that each node can compute it based purely local information and no communication is needed. To learn the local Q-tables the nodes can execute Algorithm~\ref{alg:DPQ}   in a distributed fashion. In particular, whenever node $i$ receives a routing task, or state $s=(i,j)$, it can update its Q-tables in the following steps. First node $i$ selects action $a$, corresponding to the neighboring node that will receive the packet on route to $j$ (line 10 in Algorithm~\ref{alg:DPQ}). Based on the transmission,  node $i$ receives feedback from node $j$ in terms of rewards $r^{\text{Energy}}$ and $r^{\texttt{PDR}}$ and value of the maximum Q-function from node $j$, $Q_{\beta}^a(j, a')$. 
Node $i$ can then compute the corresponding reward for each preference $\beta\in \mathcal{B}$ as $R_{\beta}=(1-\beta)r^{\text{Energy}}+\beta r^{\texttt{PDR}}$.
The node then updates its local Q-table for the corresponding state-action pair using a distributed version of the Q-learning update rule:
\begin{equation}\label{eq:distributedUpdateRule}
Q_{\beta}^i(j,a) = Q_{\beta}^i(j,a)  +  \alpha \Big(R_{\beta} +\max_{a'} Q_{\beta}^a(j, a') - Q_{\beta}^i(j,a) \Big)
\end{equation}
 This update is fully distributed, as each node relies solely on its local Q-table $Q_{\beta}^i$ and information from its neighbor node that it is transmitting to. 
 
 We illustrate the distributed algorithm in Figure~\ref{fig:distributed-Q-Learn}. In the figure,  a typical node $i$ maintains Q-table(s), storing information only about its immediate neighboring nodes $m, n, $ and $i'$, corresponding to the different preference vectors $\beta_1, \beta_2, \beta_3$. The next relay node ($i'$ in this example) is chosen based on the policy $\pi$. A data packet is sent from node $i$ to node $i'$ aimed at the destination node $j$. In our proposed approach, the chosen relay node $i'$ receives the packet and sends an acknowledgment packet back to node $i$ with the relevant information needed by $i$ to update its local Q-table using the update rule shown in equation \ref{eq:distributedUpdateRule}. This information includes the delivery reward $r^{\texttt{PDR}}(s,a)$ needed to calculate the total reward $R_\beta$, and $\max_{a'} Q_{\beta}^{i'}(j, a')$ from the Q-table of node $i'$. Thus, our proposed distributed approach adds further communication efficiency by encapsulating the relevant information in the acknowledgment packet rather than sending it through a separate packet.

\textbf{Practical Adaptation to IoT Resource Constraints:}
The distributed DPQ-learning framework can be adapted to real-world IoT deployments by selecting the preference-grid size $|\mathcal{B}|$ and training budget according to per-node memory and compute limits. Since learning is localized, each node maintains only a \emph{local} DPQ structure based on its one-hop neighborhood rather than storing Q-values for the full network. Let $|S_{\text{local}}|$ denote the number of local states represented at a node and $|A|$ be the number of candidate next-hop actions. Then, each preference point requires storing $|S_{\text{local}}|\cdot|A|$ Q-values, and the overall per-node storage is $|\mathcal{B}|\cdot|S_{\text{local}}|\cdot|A|$ floating-point values. With 64-bit floating-point storage, the memory footprint is approximately $8\,|\mathcal{B}|\,|S_{\text{local}}|\,|A|$ bytes per node, enabling a tunable accuracy--resource trade-off.

\begin{figure}[t]
     \centering \includegraphics[width=0.48\textwidth]{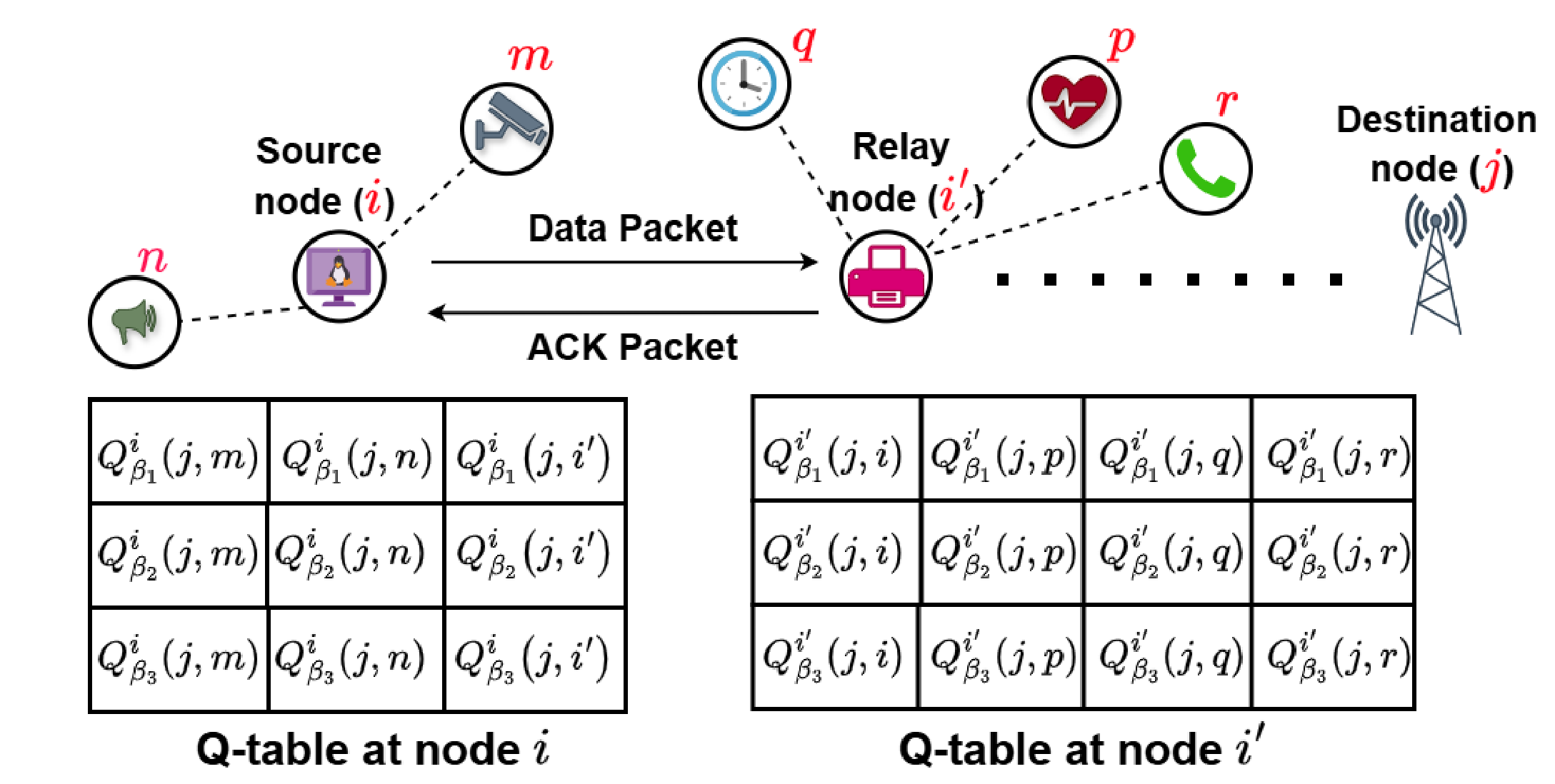}
     \caption{Depiction of Proposed Distributed DPQ-Learning Routing}
    \label{fig:distributed-Q-Learn}
\end{figure}

\section{ Performance Guarantees for the GIP}
\label{sec:theory}

In this section, we quantify the approximation accuracy of the GIP policy. As discussed in Section~IV, the learning procedure converges, at the grid points \(\beta\in\mathcal{B}\), to the corresponding optimal \(Q\)-tables \(Q_\beta\). For any intermediate \(\beta\notin\mathcal{B}\), Algorithm~2 constructs an interpolated \(Q\)-table from its two neighboring grid points and acts greedily with respect to it. We  establish a uniform approximation bound for this policy with respect to the optimal policy at the same \(\beta\). The guarantee is controlled solely by standard problem constants and the grid resolution; thus, refining the grid and improving the per-grid estimates makes the GIP policy uniformly near-optimal over the entire \(\beta\)-range.


Our bounds depend on two problem\mbox{-}level parameters. \emph{First}, the worst\mbox{-}case expected episode length under optimal control, captured by a single constant that applies to all \(\beta\) and all initial states.
Let \(K\) denote the (random) episode length, i.e., the number of time steps from an initial state \(s\) until termination.
For each \(\beta\), define the worst \(\beta\)\mbox{-}optimal expected episode length from state \(s\) by
\[
H_\beta(s)\;\triangleq\;\sup_{\pi\in\Pi_\beta^*}\,\mathbb{E}_s^{\pi}[K],
\qquad
H\;\triangleq\;\sup_{\beta\in\mathcal{B}}\;\sup_{s\in\mathcal{S}}\,H_\beta(s).
\]
We take the supremum over optimal policies $\Pi_\beta^*$ so that the bound is independent of tie\mbox{-}breaking among optimal actions. 
Under standard stochastic shortest\mbox{-}path assumptions (proper optimal policies and bounded rewards), \(H < \infty\).
%
\emph{Second}, our bounds will depend on the difference between the two objectives, or rewards. In particular, define 
\begin{equation} \label{eq:Gamma}
    \Gamma \triangleq \max_{(s,a)\in \mathcal{S}\times \mathcal{A}}|r^{\rm{Energy}}(s,a)-r^{\texttt{PDR}}(s,a)|.
\end{equation}

We now prepare to state the approximation bound for the GIP policy.
We analyze the approximation error after running Algorithm~1 and assume that the $Q$-tables learned at the grid points are accurate to within a common tolerance $\epsilon$, formalized in the following assumption.
\begin{assumption}
\label{assump:pergrid-accuracy}
For each $\beta\in\mathcal{B}$ we have an estimate $\widehat Q_{\beta}$ satisfying
\[
\|\widehat Q_{\beta}-Q_{\beta}\|_\infty \le \varepsilon,
\]
 where $Q_{\beta}$ is the optimal Q-table for preference indicator $\beta$.
\end{assumption}
Then the GIP policy for any $\beta\in [0,1]$ is a greedy policy with respect an interpolation of $\hat{Q}_{\underline{\beta}}$ and  $\hat{Q}_{\overline{\beta}}$ where  $\underline{\beta} \triangleq \max \{ b \in \mathcal{B} | b\leq \beta \}$ and $\overline{\beta}\triangleq \min \{ b \in \mathcal{B} | b\geq \beta \}$. In particular, it is greedy with respect to the Q-table
\begin{equation}\label{eq:Qint} Q_{\beta}^{\texttt{Int}}\triangleq \rho \hat{Q}_{\underline{\beta}}+ (1-\rho)\hat{Q}_{\overline{\beta}},  
\end{equation}
where $\rho \triangleq (b-\underline{\beta})/(\overline{\beta}-\underline{\beta})$.
Meaning that 
$$\pi_{\beta}^{\texttt{Greedy}}(s):=\underset{a\in \mathcal{A}}{\text{argmax}}~ Q_{\beta}^{\texttt{Int}}(s,a).$$

\begin{theorem}\label{Thm:1}
    Under Assumption~\ref{assump:pergrid-accuracy}, for any $\beta\in [0,1]$ the interpolation policy in Eq.~\ref{eq:Qint} 
    $$||Q_{\beta}^{\texttt{Int}}-Q_{\beta}||_{\infty}\leq \epsilon+\Gamma (H+1) (\overline{\beta}-\underline{\beta})$$
    where  $\underline{\beta} \triangleq \max \{ b \in \mathcal{B} | b\leq \beta \}$ and $\overline{\beta}\triangleq \min \{ b \in \mathcal{B} | b\geq \beta \}$.
\end{theorem}

What this bound means is: if the $Q$-learning in Algorithm~1 runs at the grid points and produces approximate estimates of the optimal $Q$-tables, which is ensured under standard conditions (e.g., sufficient exploration and appropriate stepsizes), then for any preference $\beta$ the policy that is greedy with respect to the interpolated table is guaranteed to be near–optimal. In particular, its deviation from the optimal $Q_\beta$ is uniformly controlled across all $\beta\in[0,1]$ by the sum of the per–grid learning accuracy $\epsilon$ and a term that depends only on the local grid gap $(\overline{\beta}-\underline{\beta})$ and fixed problem constants. Practically, this provides a simple design rule: choose the grid resolution to meet a target accuracy, trading off computational cost (number of grid points) against uniform suboptimality of the GIP policy.

\subsection{Proof of Theorem~\ref{Thm:1}}

 To prove Theorem~\ref{Thm:1}, we first establish a basic regularity property: the mapping $\beta \mapsto Q_{\beta}$ is \emph{Lipschitz continuous}.
\begin{lemma}
\label{lem:lipschitz-Q}
The optimal $Q$-table varies at most linearly with~$\beta$:
\[
\|Q_{\beta}-Q_{\beta'}\|_{\infty} \;\le\; \Gamma\,(H+1)\,|\beta-\beta'|\qquad \text{for all } \beta,\beta'\in[0,1],
\]
where $\Gamma$ is the reward-variation constant and $H$ is the worst-case $\beta$-optimal expected episode length.
\end{lemma}
\begin{proof}
    Recall that the scalarized reward for  $\beta\in[0,1]$ is a convex combination
\[
r_\beta(s,a)=\beta r^{\rm{Energy}}(s,a)+(1-\beta) r^{\texttt{PDR}}(s,a).
\]
Then, recalling the definition of $\Gamma$ in Eq.~\eqref{eq:Gamma}, for all $(s,a)$ and all $\beta,\beta'\in[0,1]$,
\begin{align*}
|r_\beta(s,a)-r_{\beta'}(s,a)|
=& |\beta-\beta'|\;|r^{\rm{Energy}}(s,a)-r^{\texttt{PDR}}(s,a)| \\
\le& \Gamma\,|\beta-\beta'|.
\end{align*}


Fix a stationary policy $\pi$ and an initial state $s$. Let $K$ denote (random) episode length until termination. We have that
\begin{align*}
\big|V_\beta^\pi(s)-V_{\beta'}^\pi(s)\big|
=& \Big|\mathbb{E}_s^\pi\Big[\sum_{t=0}^{K-1}\big(r_\beta(s_t,a_t)-r_{\beta'}(s_t,a_t)\big)\Big]\Big| \\
\le& \Gamma\,|\beta-\beta'|\,\mathbb{E}_s^\pi[K].
\end{align*}
Let $\pi_\beta\in\Pi_\beta^*$ be $\beta$-optimal. Then, for any $s$,
\begin{align*}
V_\beta(s)-V_{\beta'}(s)
\le& V_\beta^{\pi_\beta}(s)-V_{\beta'}^{\pi_\beta}(s) 
\le \Gamma\,|\beta-\beta'|\,\mathbb{E}_s^{\pi_\beta}[K] \\
\le& \Gamma\,H\,|\beta-\beta'|.
\end{align*}
Swapping $\beta$ and $\beta'$ gives the reverse inequality; hence
\[
\|V_\beta-V_{\beta'}\|_\infty \le \Gamma\,H\,|\beta-\beta'|.
\tag{3}
\]

Finally, for any $(s,a)$ and random next state $s'\sim T(s,a)$ 
\[
\begin{aligned}
|Q_\beta(s,a)-Q_{\beta'}(s,a)|
&\leq \big|r_\beta(s,a)-r_{\beta'}(s,a)\big| \\
  &~~~~~~~~~~ + \Big|\mathbb{E}\big[V_\beta(s')-V_{\beta'}(s')\mid s,a\big]\Big| \\
&\le \Gamma\,|\beta-\beta'| + \|V_\beta-V_{\beta'}\|_\infty \\
&\le \Gamma\,(1+H)\,|\beta-\beta'|.
\end{aligned}
\]
Taking the maximum over $(s,a)$ yields
\[
\|Q_\beta^*-Q_{\beta'}^*\|_\infty \le \Gamma\,(1+H)\,|\beta-\beta'|.
\]
\end{proof}
We are now ready to prove Theorem~\ref{Thm:1}.  By the triangle inequality, we can slit the error in the following two terms:
\begin{align*}
    ||Q_{\beta}^{\texttt{Int}}-Q_{\beta}||_{\infty} \leq&
\underbrace{\big\|\rho\big(\widehat Q_{\underline{\beta}}-Q_{\underline{\beta}}\big)
+(1-\rho)\big(\widehat Q_{\overline{\beta}}-Q_{\overline{\beta}}\big)\big\|_{\infty}}_{T_1} \\
&~\;+\;
\underbrace{\big\|\rho\,Q_{\underline{\beta}}+(1-\rho)\,Q_{\overline{\beta}}-Q_{\beta}\big\|_{\infty}}_{T_2}.
 %
\end{align*}
To bound $T_1$ we use triangle inequality and Assumption~\ref{assump:pergrid-accuracy} to obtain
$
T_1 \;\le\;\rho\,\varepsilon+(1-\rho)\,\varepsilon\;=\;\varepsilon.$
To bound $T_2$, we note that by Lemma~\ref{lem:lipschitz-Q}, $\|Q_{\underline{\beta}}-Q_{\beta}\|_\infty \le \Gamma (H+1)\,|\beta-\underline{\beta}|$ and $\|Q_{\overline{\beta}}-Q_{\beta}\|_\infty \le \Gamma (H+1)\,|\overline{\beta}-\beta|$. 
Hence,
\begin{align*}
T_2 & \;\le\;\rho\,\Gamma (H+1)\,|\beta-\underline{\beta}|+(1-\rho)\,\Gamma (H+1)\,|\overline{\beta}-\beta|
\\ & \;\le\;\Gamma (H+1)\,(\overline{\beta}-\underline{\beta}).
\end{align*}
Combining these results, yields the result
$||Q_{\beta}^{\texttt{Int}}-Q_{\beta}||_{\infty}\leq T_1+T_2\leq \epsilon + \Gamma (H+1)\,(\overline{\beta}-\underline{\beta})$

\color{black}

\subsection{Limitations and Complexity Analysis}

The DPQ-Learning framework provides strong adaptability but introduces several computational and practical trade-offs:

\begin{itemize}
    \item \textbf{Memory complexity: $O(n \cdot |S| \cdot |A|)$} \\
    The algorithm maintains one Q-table for each preference grid value, leading to memory usage that grows linearly with the number of grid points $n = |\mathcal{B}|$. This is manageable for moderate $n$ but may limit deployment on ultra-low-power IoT nodes.

    \item \textbf{Per-step update complexity: $O(n)$} \\
    At each interaction step, all $n$ Q-tables are updated in parallel using the same transition sample. Although each update is lightweight, the total per-step cost increases linearly with $n$.

    \item \textbf{Action-selection overhead: $O(|A|)$} \\
    The greedy interpolation policy requires evaluating two Q-values per action from neighboring grid points. Since interpolation is constant-time, the action-selection cost remains $O(|A|)$, similar to single-objective Q-learning.

    \item \textbf{Dependence on grid resolution} \\
    Finer grids improve interpolation fidelity but incur higher memory and update costs, whereas coarser grids reduce resource usage at the expense of approximation error. This creates a tunable accuracy--complexity trade-off.
\end{itemize}

These considerations emphasize that the grid size and learning configuration should be tuned according to the computational constraints of the target IoT deployment. It should be noted that in reinforcement learning—particularly in multi-objective settings—the total runtime is inherently stochastic and depends on environment interactions, episode lengths, and convergence behavior. Thus, a closed-form time-complexity expression for the overall learning process is not provided. As already captured in the above discussion through memory requirements, per-step update costs, action-selection overhead, and grid-resolution trade-offs, computational performance is therefore best characterized via empirical runtime and sample efficiency rather than analytical bounds.

\color{black}

\section{Experimental Results}
\label{sec:Simulations}

We compare the efficiency
of the proposed Distributed DPQ-Learning algorithm with the following baselines: 
\begin{enumerate}
    \item Reinforcement-Learning-Based Energy
Efficient Control and Routing Protocol for
Wireless Sensor Networks \textbf{  (\texttt{RLBEEP})} \cite{abadi2022rlbeep}
    
    \item Reinforcement Learning for Life Time Optimization \textbf{  (\texttt{R2LTO})} \cite{bouzid2020efficient}
    
    \item Reinforcement Learning Based Routing \textbf{ (\texttt{RLBR})} \cite{guo2019optimizing}

    \item Fuzzy multiobjective routing for maximum lifetime and minimum delay \textbf{
(\texttt{FMOLD})} \cite{5158843}

    \item Recursive Shortest Path Algorithm based Routing \textbf{ (\texttt{RSPAR})} \cite{7435250}

   \item Static MultiObjective Reinforcement-Learning based Routing \textbf{  (\texttt{SMORLR})} 
    
\end{enumerate}

\texttt{SMORLR} is the brute force version of MORL-based routing with the same objectives as the proposed method. But, it starts learning from scratch when new changes occur in preferences.

\subsection{The Simulation Setup and Parameters}

The experiments are conducted using a custom Python-based simulator (Python~3.9.13) executed on an 11th Gen Intel(R) Core(TM) i7-1185G7 CPU @ 3.00GHz with 32~GB RAM on Windows. The simulated network follows an IEEE~802.15.4 MAC abstraction over a $10\times10$ grid topology ($100$ nodes), where each grid cell corresponds to an IoT node and edges represent one-hop connectivity. This corresponds to a unit-disk communication model where each node can directly communicate with its four adjacent neighbors (one-hop radio range). With the adopted local DPQ representation, the simulated per-node routing memory is approximately $1.72$~KB for the fine preference grid ($|\mathcal{B}|=11$) and $320$~bytes for the coarse grid ($|\mathcal{B}|=2$), assuming 64-bit floating-point Q-values.

The sink is fixed at the bottom-right node, while the packet source node is randomly selected at the beginning of each episode. Each episode corresponds to routing a single packet to the sink, and routing terminates either when the sink is reached or when the packet enters a pre-defined unreliable node. Unreliable nodes emulate forwarding failures by dropping the packet with probability $p_{\text{drop}}$, consistent across all methods. The data packet size is $133$~Bytes and per-packet overhead is $30$~Bytes. Energy consumption is modeled by subtracting $E_{idle}=0.5$~nJ per time-step for all nodes (idle/active cost) and $E_{tx}=0.007$~mJ for each forwarding hop, with initial node energy set to $25.0$~mJ. Additional parameters are set to $\alpha=0.9$, and $E_{awake}=0.5\,\mu$J. 

We evaluate four settings defined by two orthogonal factors: \emph{(i)} exploration strategy and \emph{(ii)} preference-change frequency, as summarized in Fig.~\ref{fig:Exp_setting_dynamic_online}. 
For exploration, Experiments~1--2 use \emph{sequential exploration/exploitation}, where the proposed algorithm applies $\epsilon=1$ for the first $1000$ episodes (exploration) and $\epsilon=0$ thereafter (exploitation). Experiments~3--4 use \emph{simultaneous exploration/exploitation}, where $\epsilon$ is initialized at $1$ and linearly decayed to $0$ over time.
For preference variation, Experiments~1 and~3 update the preference periodically every $1000^{\text{th}}$ episode (periodically changing objectives), while Experiments~2 and~4 update the preference at every episode (rapidly changing objectives). 
For the \texttt{SMORLR} baseline, $\epsilon$ is reset to $1$ after each preference update and then linearly decayed to $0$, enabling fresh exploration under the new objective.



\begin{figure}[h]
\begin{subfigure}[b]{0.24\textwidth} 
         \centering
    \includegraphics[width=\textwidth]{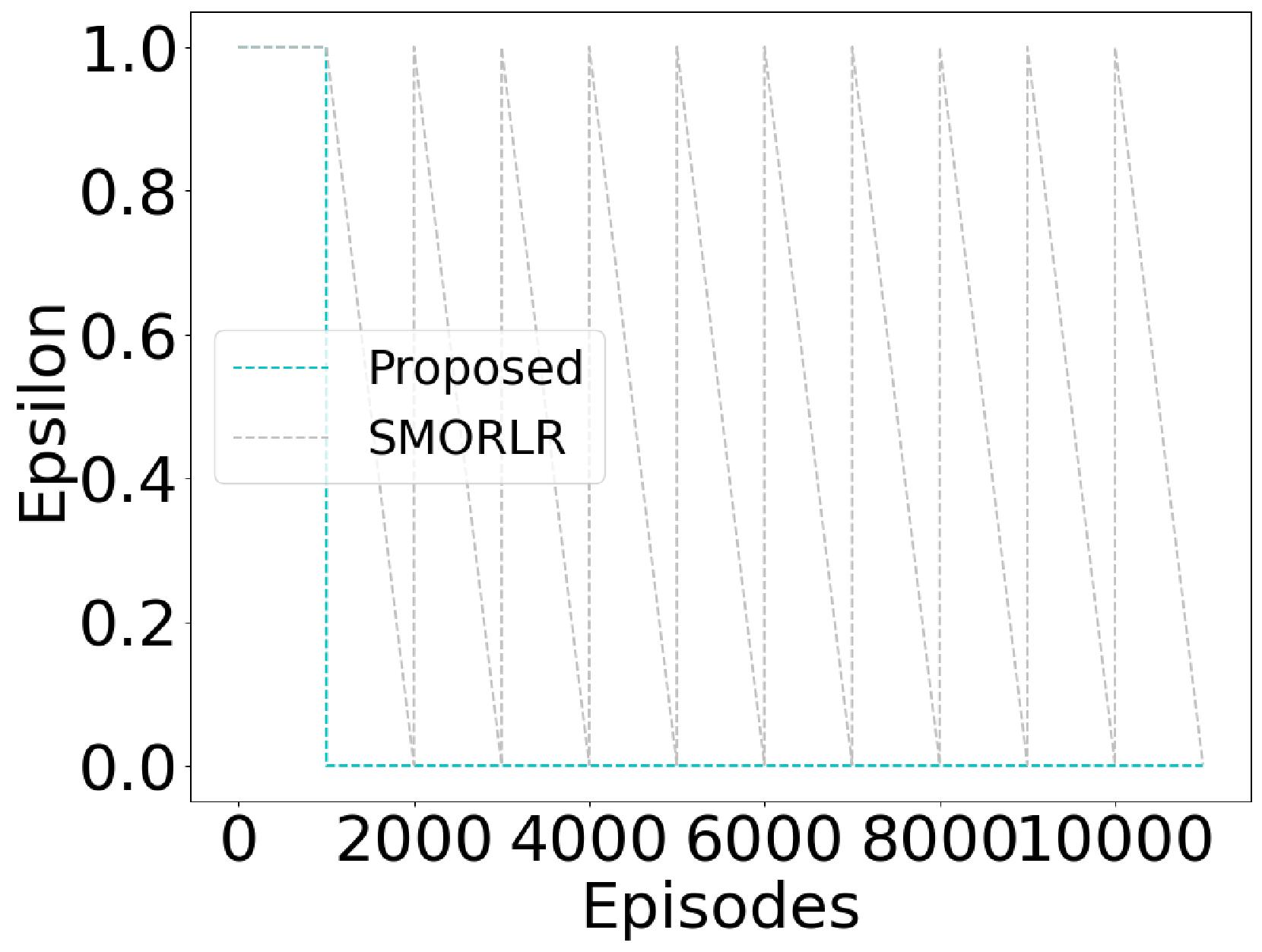}
\caption{Epsilon variation in experiments 1 and 2}\label{fig:epsilonVariation1}
\end{subfigure}
      \hfill
       \begin{subfigure}[b]{0.24\textwidth} 
         \centering
         \includegraphics[width=\textwidth]{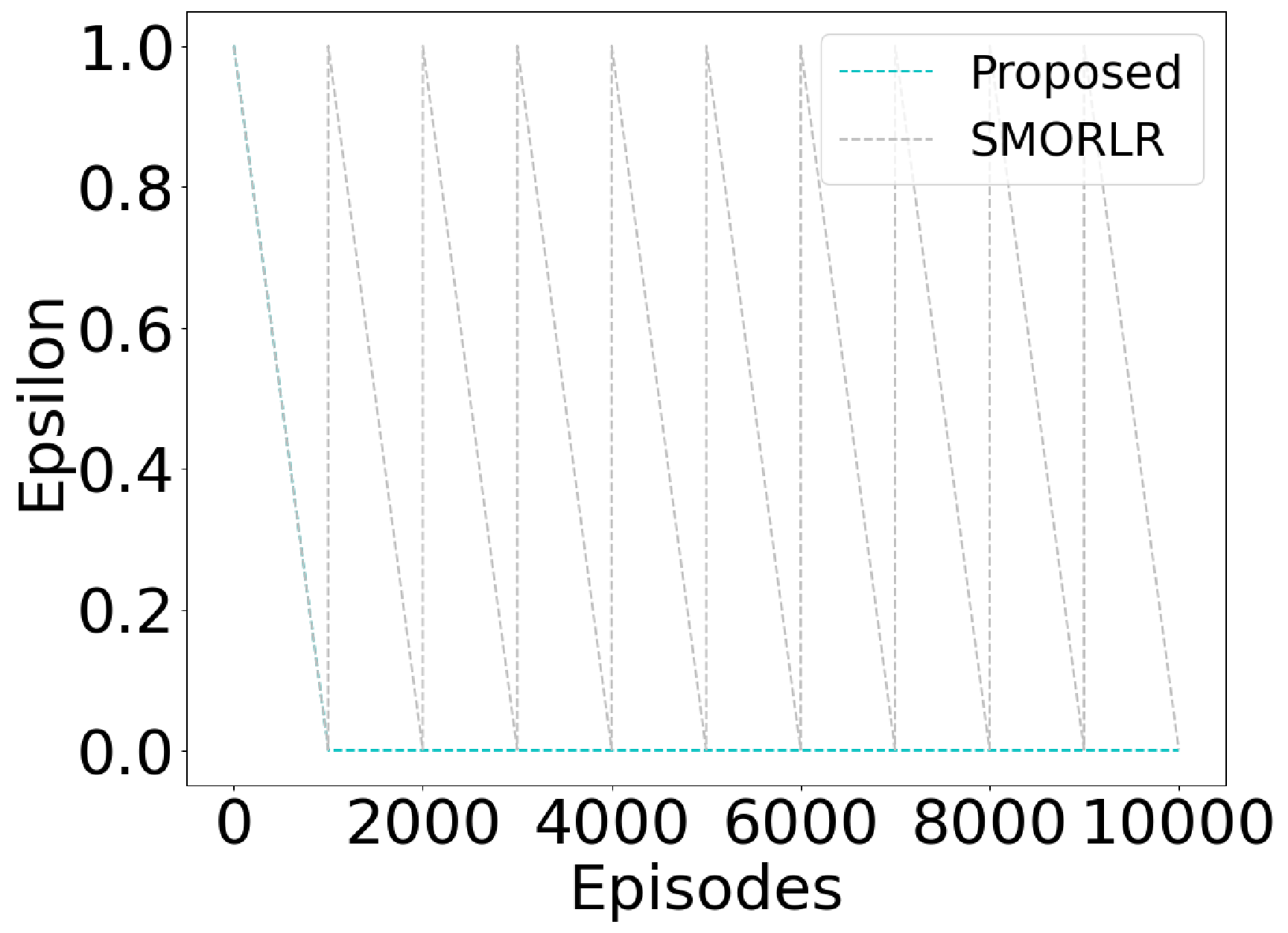}
\caption{Epsilon variation in experiments 3 and 4}\label{fig:epsilonVariation2}
     \end{subfigure}
      \vfill
     \begin{subfigure}[b]{0.24\textwidth} 
         \centering
         \includegraphics[width=\textwidth]{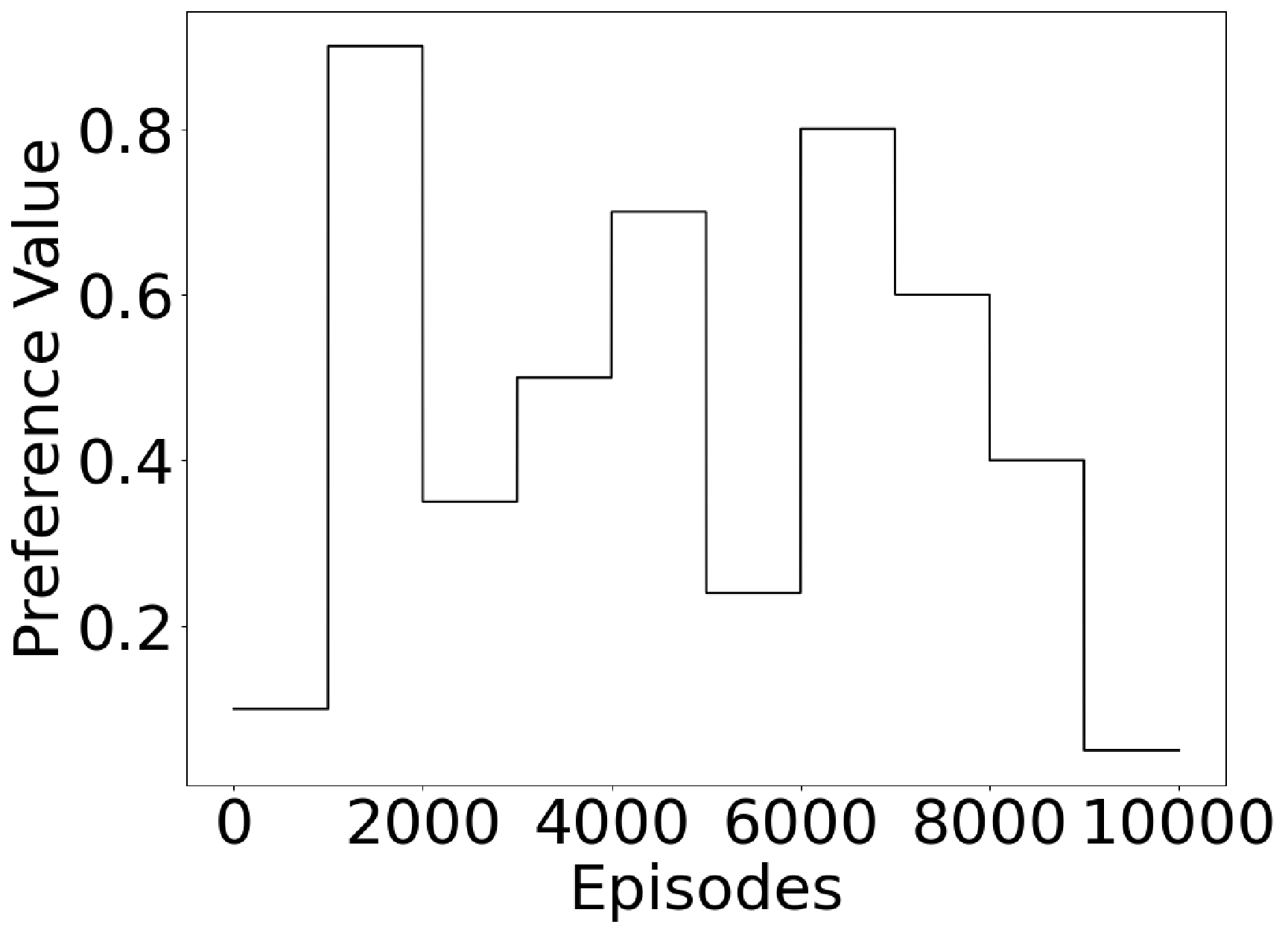}
\caption{Preference Variation in experiments 1 and 3}
\label{fig:preferenceVariation1}
     \end{subfigure}
     \hfill
     \begin{subfigure}[b]{0.24\textwidth} 
         \centering
         \includegraphics[width=\textwidth]{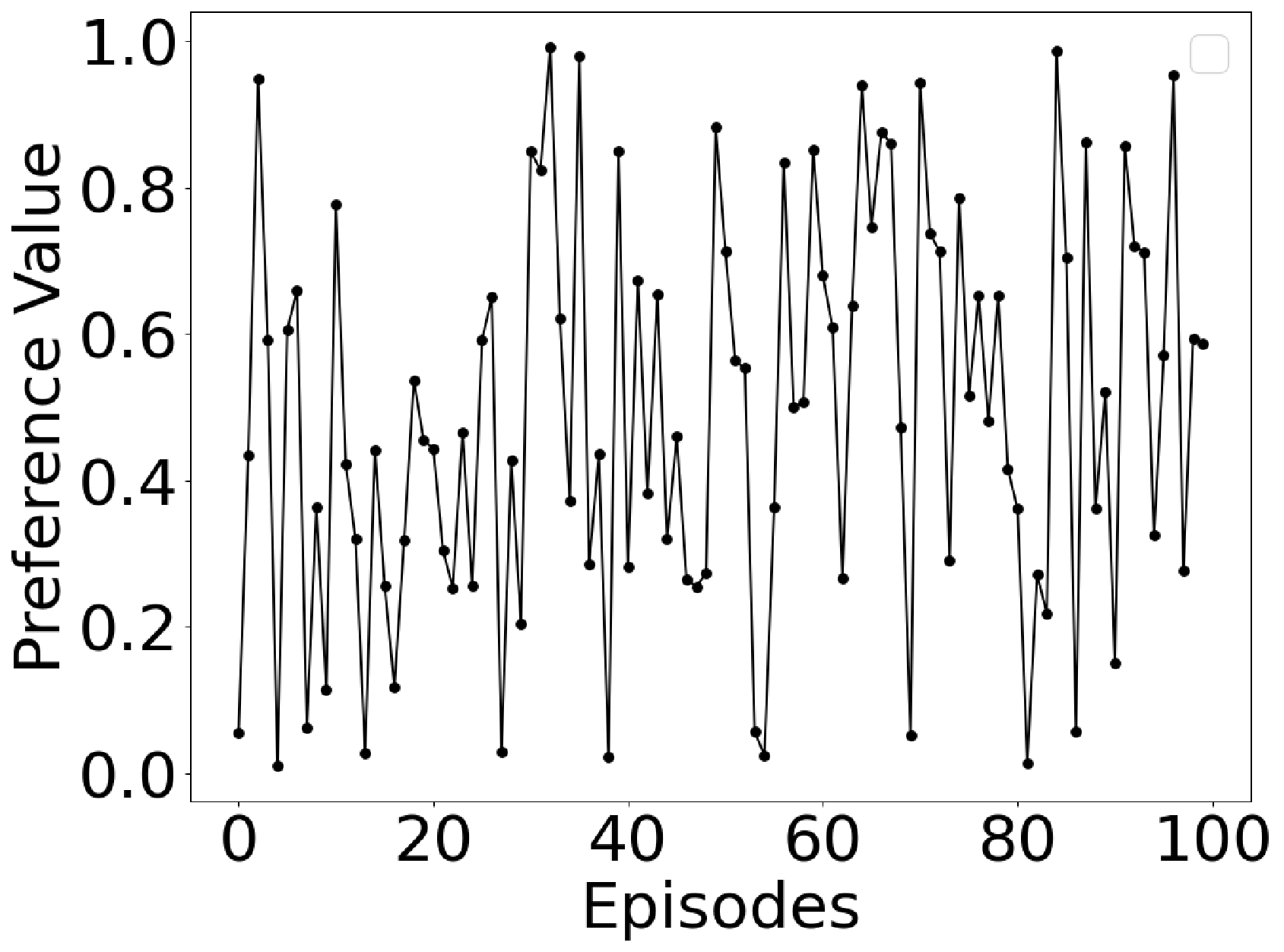}
\caption{Preference Variation in experiments 2 and 4}
\label{fig:preferenceVariation2}
     \end{subfigure}
   \caption{Variation of Epsilon and Preferences during the experimental simulations.}
        \label{fig:Exp_setting_dynamic_online}
\end{figure}

In the simulations, transmission priorities are represented by the preference parameter $\beta$, which adjusts the balance between packet delivery ratio (PDR) and energy efficiency. Lower $\beta$ values (e.g., $\beta = 0.1$) correspond to high-priority transmissions that favor higher PDR and reliability, even at the cost of increased energy consumption. In contrast, higher $\beta$ values (e.g., $\beta = 0.9$) represent low-priority transmissions that prioritize energy savings over delivery reliability. We compare the performance of the proposed Distributed DPQ-Learning routing algorithm with the six baselines on three vital metrics: overall reward, packet delivery, and energy consumption. The following subsections provide detailed analysis of these three metrics under sequential and simultaneous exploration-exploitation scenarios. 

\subsection{Sequential Exploration and Exploitation}\label{dynamic-offline}
We provide a detailed performance analysis of the proposed work based on a sequential exploration and exploitation scheme with varying frequency of preference variation at every thousandth episode (in Experiment 1) and every episode (in Experiment 2). In the below subsections, we compare the performance of the proposed work and baseline algorithms using overall reward, packet delivery, and energy consumption.

\subsubsection{Overall Reward}\label{dynamic-offline-reward}
The overall reward captures the optimization goal by comprising both objectives as per the current episode's preference $\beta_m$. At an episode $m$, when the state is $s$, and a routing action $a$ is taken, the overall reward is $r^{\texttt{Overall}}$ is defined as:
\begin{equation}\label{eq:overallReward}
r^{\texttt{Overall}}(s, a)= \beta_m r^{\texttt{Energy}}(s, a)  +(1-\beta_m)r^{\texttt{PDR}}(s, a)
\end{equation}

\begin{figure*}[!t]
     \begin{subfigure}[b]{0.32\textwidth} 
         \centering
         \includegraphics[width=\textwidth]{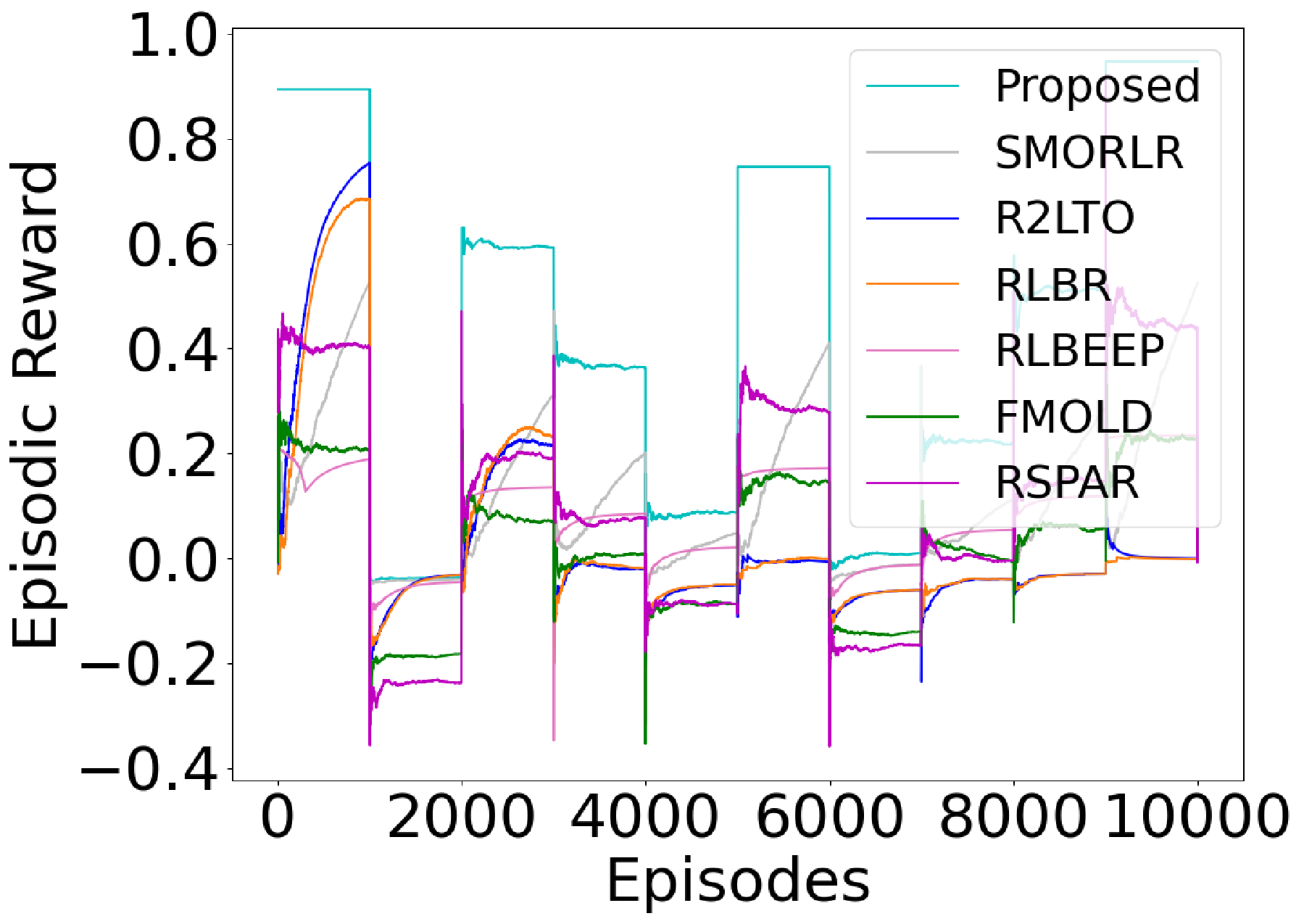}
\caption{Experiment 1: Comparison of the episodic rewards}\label{fig:reward_dynamic_offline}
     \end{subfigure}
     \hfill
      \begin{subfigure}[b]{0.32\textwidth} 
         \centering
         \includegraphics[width=\textwidth]{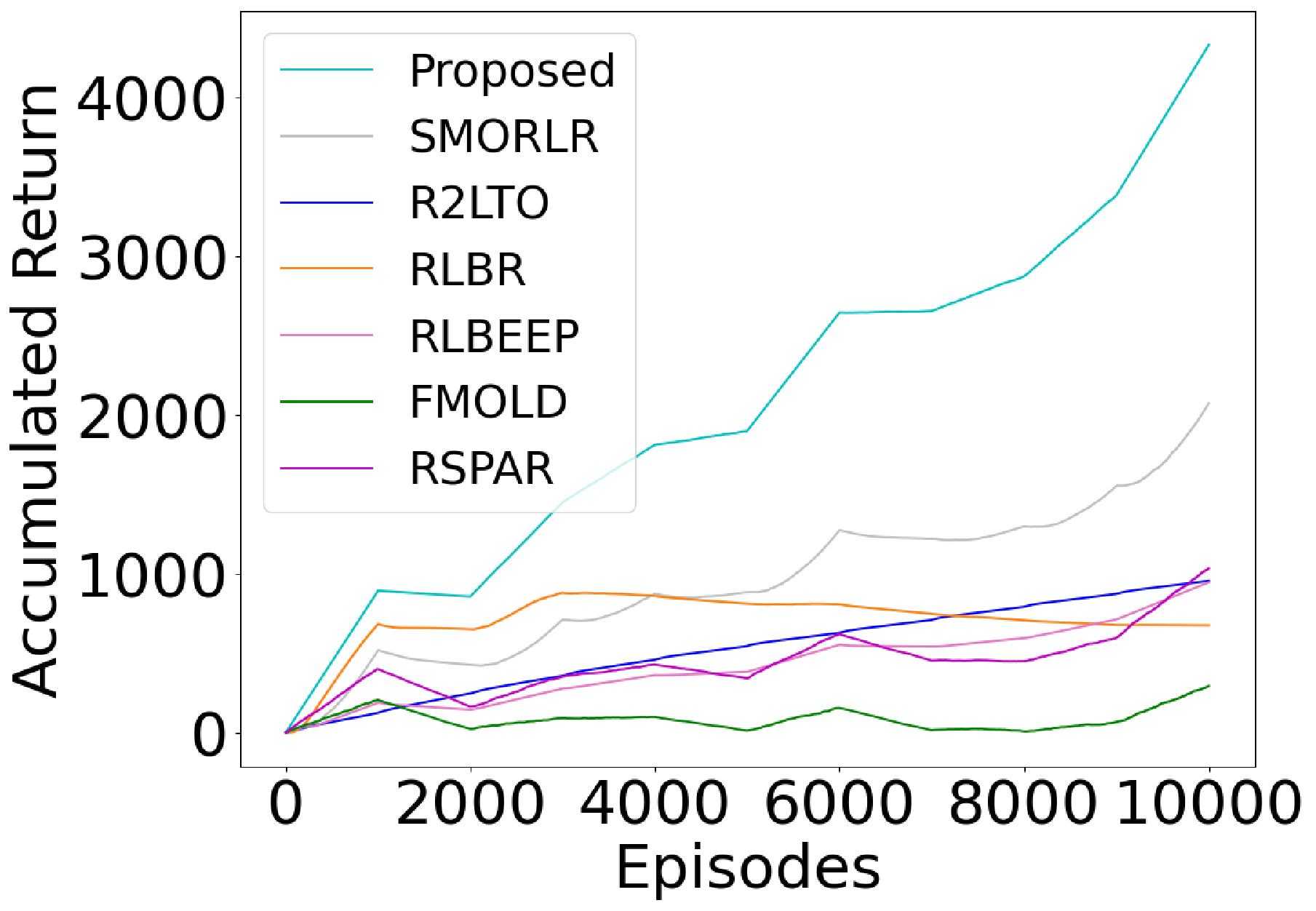}
\caption{Experiment 1: Comparison of accumulated returns}\label{fig:accumulated_offline}
     \end{subfigure}
      \hfill
\begin{subfigure}[b]{0.32\textwidth}
         \centering
         \includegraphics[width=\textwidth]{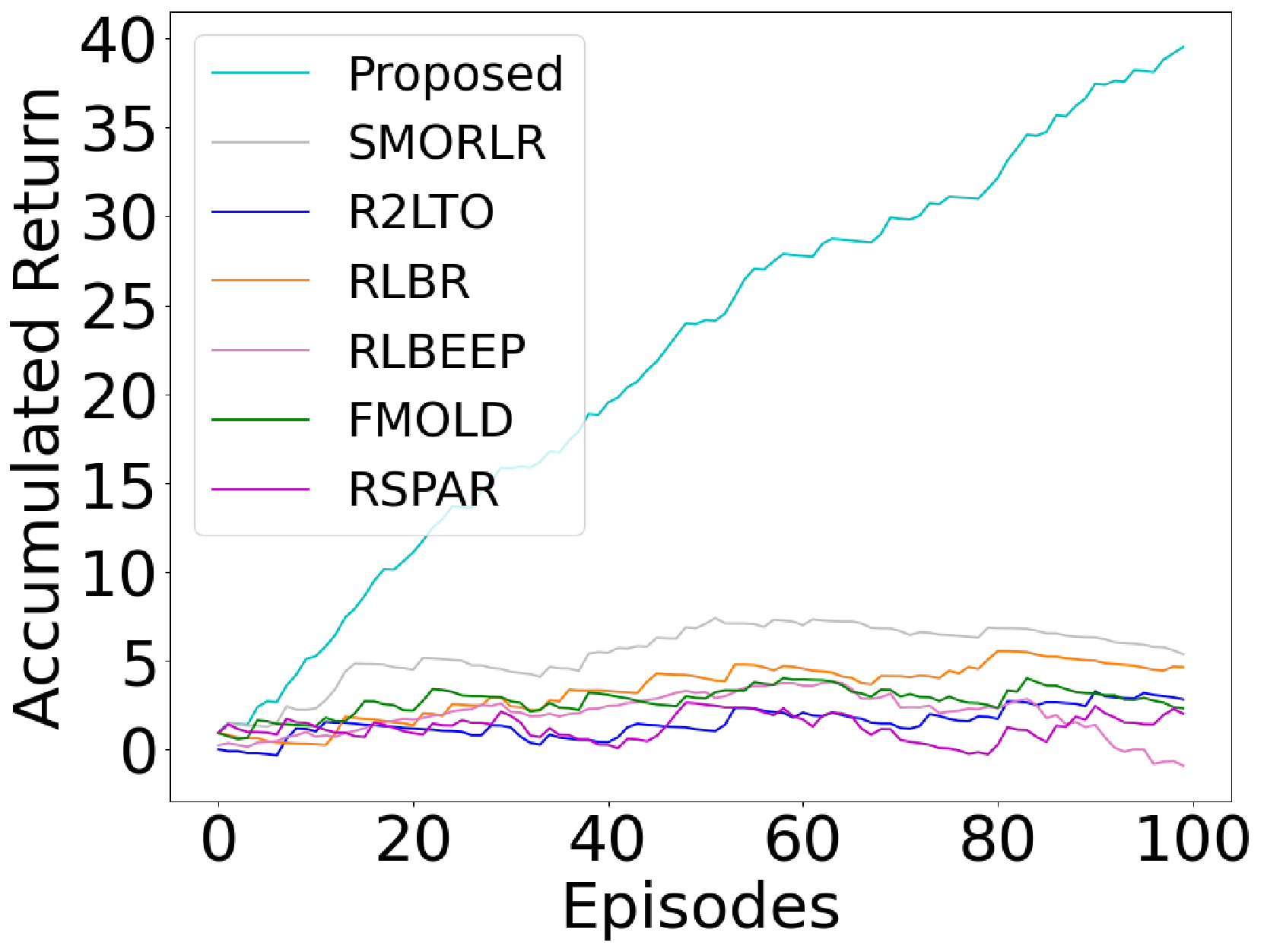}
         \caption{Experiment 2: Accumulated Return v/s Episodes}
         \label{fig:accumulated_Return_offline_cont_change}
     \end{subfigure}
   \caption{Overall reward in Sequential exploration-exploitation scheme.}
        \label{fig:Results_Overall_reward_Sequential}
\end{figure*}

Figure~\ref{fig:Results_Overall_reward_Sequential} shows the performance comparison of $r^{\texttt{Overall}}(s, a)$ in the sequential exploration-exploitation scheme between the proposed and existing works. Figure \ref{fig:reward_dynamic_offline} and Figure \ref{fig:accumulated_offline} shows experiment 1's episodic and accumulated rewards, respectively. 
Figure \ref{fig:reward_dynamic_offline} shows that the proposed Distributed DPQ-Learning routing algorithm immediately adapts according to the preference and gives high performance even if the preferences change. This superior performance is attributed to the Pareto front, which is learned during the exploration phase. In contrast \texttt{SMORLR} starts learning the policy suited to the new preference after each preference shift. Although it focuses on the same objectives, it still does not learn enough to outperform the proposed dynamic protocol after 1000 episodes of learning. Static protocols perform poorly in environments where preferences change dynamically. Algorithms like \texttt{R2LTO}, \texttt{RLBR}, and \texttt{RLBEEP}, which are primarily energy-focused and static, display performance that is relatively similar to one another for most preferences. However, they underperform by an average of 102\% compared to the proposed approach. During the intervals between $1000-2000$ and $6000-7000$ episodes, these algorithms begin to show improvement, as the preference weight, $\beta$, shifts to $0.1$ and $0.2$, giving higher priority to energy efficiency--an objective for which they are optimized. \texttt{RSPAR} underperforms by 263\% on average compared to the proposed method, while \texttt{FMOLD}, despite being a multi-objective approach, is 266\% less efficient. Moreover, they don't seem to be learning efficient policies with time. However, the proposed algorithm performs better immediately (from the beginning of these intervals) because it utilizes the Pareto front learned during exploration. 

The cumulative reward is shown in Figure \ref{fig:accumulated_offline}, and confirm superiority in the proposed distributed DPQ-Learning algorithm over baseline approaches over the long run. Through the episodes, this outperformance becomes increasingly evident. Finally, after 10000 episodes, we can notice that the accumulated return for the proposed method is better than \texttt{SMORLR} by $2.3\times$, \texttt{RLBEEP} by $5.12\times$, \texttt{R2LTO} and \texttt{FMOLD} by $5.25\times$, \texttt{RLBR} by $6\times$, and \texttt{RSPAR} by $16.8\times$. Similarly, in experiment 2, preferences were randomly changed in each episode and the measured performance is shown in Figure~\ref{fig:accumulated_Return_offline_cont_change}. Here we noticed that the accumulated returns of the proposed routing protocol are consistently higher than baselines, and the difference becomes larger with more episodes. At the end of $100$ episodes, the accumulated return of the proposed algorithm is around $43$ which is $3.3\times$ of \texttt{RSPAR}, $4.3\times$ of \texttt{SMORLR}, $6.14\times$ of \texttt{RLBR}, $7.16\times$ of \texttt{R2LTO}, $8.6\times$ of \texttt{FMOLD} and $43\times$ of \texttt{RLBEEP}. Our observations confirm the proposed method consistently outperforms all baselines and adapts to this efficient performance irrespective of preference changes. 


\subsubsection{Energy}\label{dynamic-offline-energy}
\begin{figure*}[!t]
  \begin{subfigure}[b]{0.32\textwidth} 
         \centering
         \includegraphics[width=\textwidth]{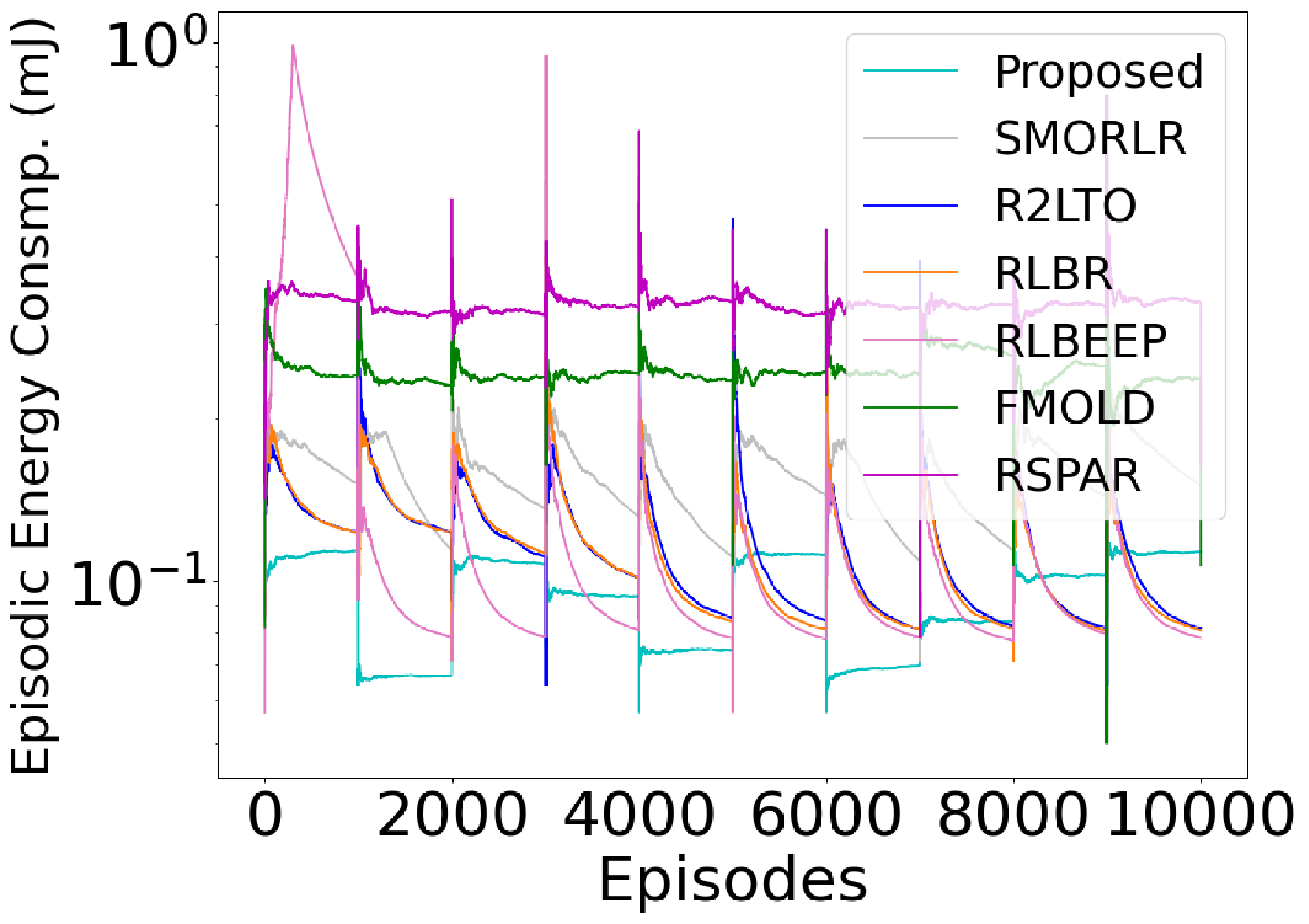}
         \caption{Experiment 1: Episodic Energy Consumption}
         \label{fig:energy_dynamic_offline}
     \end{subfigure}
      \hfill
    \begin{subfigure}[b]{0.32\textwidth} 
         \centering
         \includegraphics[width=\textwidth]{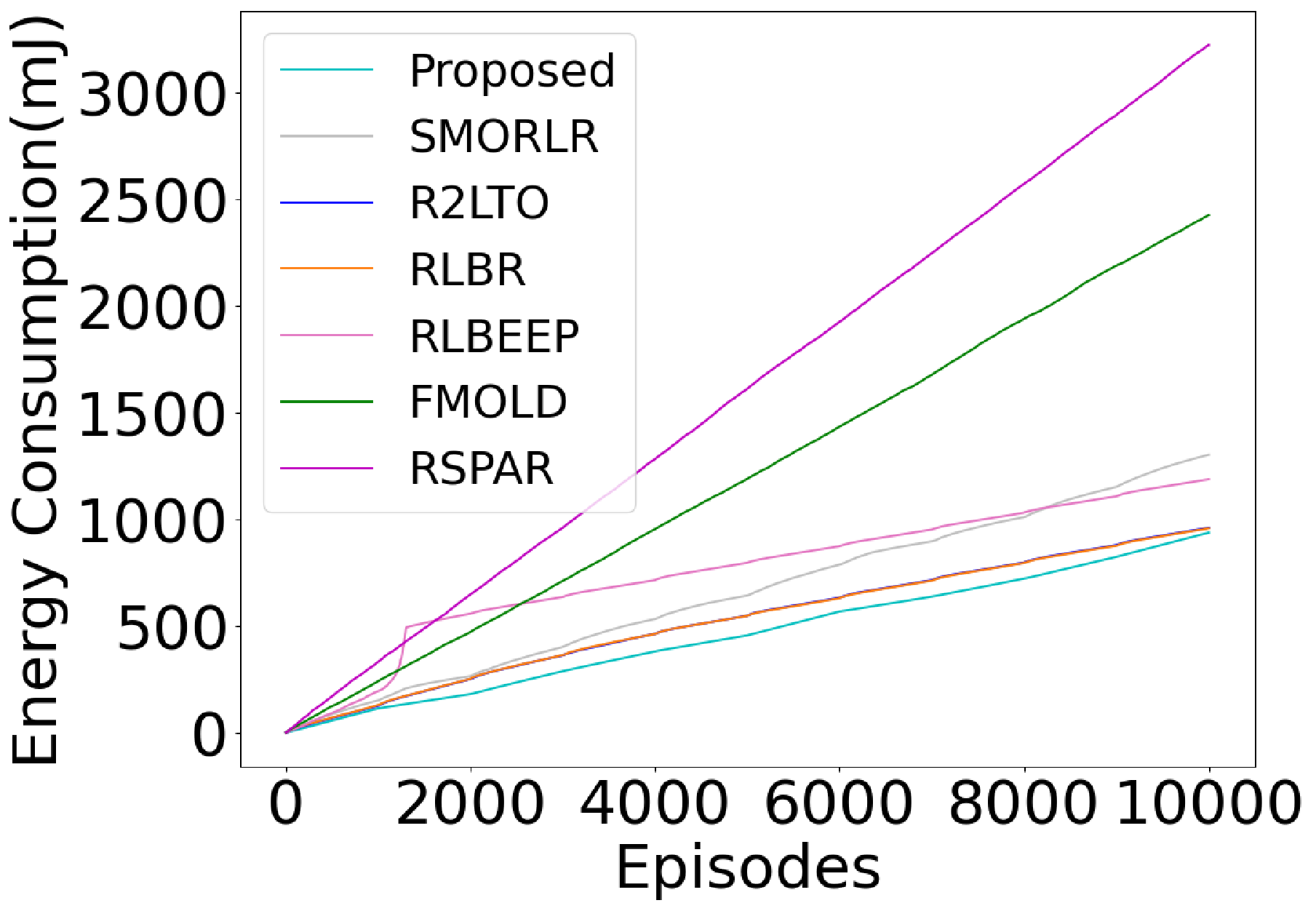}
\caption{Experiment 1: Comparison of total energy consumption} \label{fig:offline_total_Energy}
     \end{subfigure}
     \hfill
     \begin{subfigure}[b]{0.32\textwidth} 
         \centering
         \includegraphics[width=\textwidth]{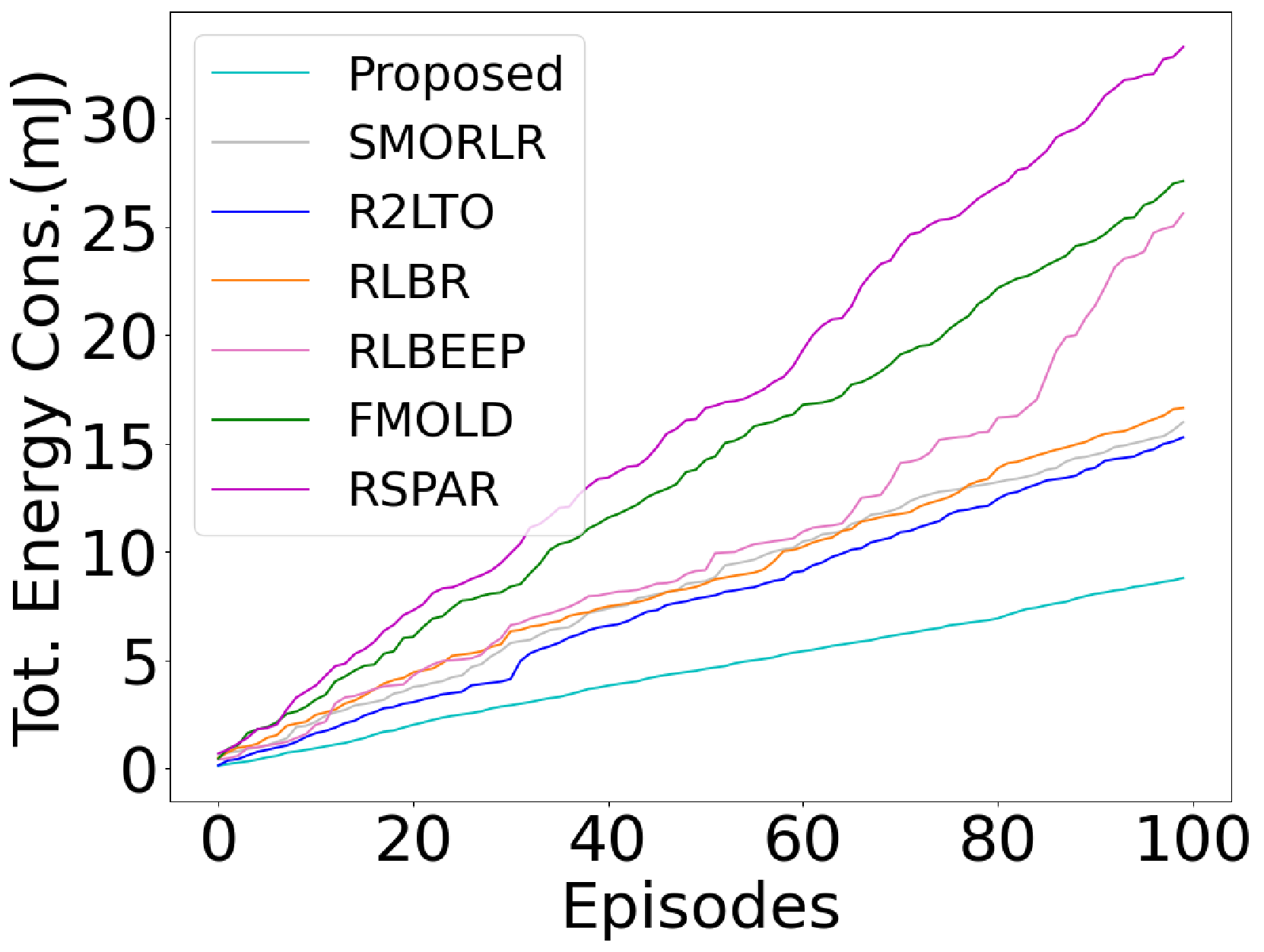}
         \caption{Experiment 2: Total Energy consumption v/s Episodes}
         \label{fig:total_energy_offline_cont_change}
     \end{subfigure}
   \caption{Energy Consumption in the Sequential exploration-exploitation scheme.}
        \label{fig:Results_Energy_Concumption}
\end{figure*}
Figure~\ref{fig:energy_dynamic_offline} compares episodic energy consumption (in mJ) across baseline algorithms over $10,000$ episodes, showing their performance trends. The Y-axis in Figure~\ref{fig:energy_dynamic_offline} is logarithmic, ranging from $0$ to $1$ mJ. The proposed distributed DPQ-Learning routing protocol leads to minimal episodic energy consumption whenever a higher preference is given to the energy objective while performing reasonably well even when the preference for the energy objective is low. On average, the proposed distributed DPQ-Learning routing protocol shows about $80-90$\% lower energy consumption than baseline algorithms. \texttt{SMORLR} consumes around $150-200\%$ more energy than the proposed algorithm over time. However, with each preference change, it starts to learn to converge closer to the proposed algorithm gradually. However, the more often preferences change, the inefficient this approach becomes. \texttt{R2LTO} and \texttt{RLBR} tend to stabilize close to $0.3$ mJ of episodic energy consumption, with about $50\%$ higher consumption than the proposed method. \texttt{RLBEEP} spikes at first and then stabilizes around $0.4$ mJ. It has approximately $100\%$ higher consumption than the proposed algorithm. \texttt{FMOLD} stabilizes near $0.6$ mJ, about $3$ times the energy of the distributed DPQ-Learning approach. The energy consumption is the highest among the four, but the consistency is better than \texttt{RSPAR}. The lowest energy efficiency can be seen in \texttt{RSPAR}, which spikes above $1 $ mJ initially and stabilizes around $0.8$ mJ. Its energy consumption is $4\times$ to $5\times$ higher than the proposed algorithm. Although \texttt{RLBEEP}, \texttt{R2LTO}, and \texttt{RLBR} learn to perform well or better than the proposed method for certain sets of preferences, but their overall performance over a long period is lesser than the proposed method.

The cumulative effects over the long run are shown in Figure~\ref{fig:offline_total_Energy}, which depicts cumulative energy consumption over episodes. The figure shows that the proposed algorithm is the most energy-efficient, consuming around $4\times$ and $3.4\times$ less energy than \texttt{RSPAR} and \texttt{FMOLD}, respectively, $30$--$50\%$ less than \texttt{SMORLR} and \texttt{RLBEEP}, and $3$--$5\%$ less than \texttt{RLBR} and \texttt{R2LTO} by episode $10000$. While \texttt{SMORLR}, \texttt{RLBEEP}, \texttt{R2LTO}, and \texttt{RLBR} perform moderately well, \texttt{RSPAR} and \texttt{FMOLD} exhibit the highest cumulative energy consumption, highlighting their inefficiency. Nonetheless, the proposed routing algorithm consumes less energy than all baselines. We further evaluate Experiment~2, where preferences are randomly changed each episode. As shown in Figure~\ref{fig:total_energy_offline_cont_change}, the proposed protocol again achieves lower cumulative energy consumption, with the gap widening over time. By the end of 100 episodes, the cumulative energy consumption of \texttt{RSPAR} and \texttt{RLBEEP} is $3.87\times$, \texttt{FMOLD} is $3.53\times$, \texttt{RLBR} is $2\times$, \texttt{R2LTO} is $1.73\times$, and \texttt{SMORLR} is $1.67\times$ that of the proposed routing scheme. Therefore, regardless of exploration policy or preference variation pattern, the proposed method results in substantially lower energy consumption than the baseline algorithms.

\subsubsection{Packet Delivery}\label{dynamic-offline-pdr}
Figure \ref{fig:delivery_dynamic_offline} and Figure \ref{fig:total_delivery_offline} depict the variation in episodic and cumulative packet deliveries, respectively. Figure \ref{fig:delivery_dynamic_offline} shows that the proposed distributed DPQ-Learning routing algorithm outperforms all the baselines for episodes ($0-1000, 2000-4000, 5000-6000, 7000-10000$) where the preference for the \texttt{PDR} objective is more than 0.5. \texttt{RSPAR} and \texttt{FMOLD} tend to hover around episodic \texttt{PDR} of $0.4-0.5$ and $0.3-0.4$ respectively. \texttt{SMORLR} tends to learn with time to maximize the \texttt{PDR} for each new preference vector. However, its peak performances are $50\%$ lower than the proposed method's consistent performances. \texttt{RLBEEP}, \texttt{R2LTO}, and \texttt{RLBR} focus on the energy aspect of the system and did not consider \texttt{PDR}. As a result, they consistently ignore this objective and performing worse. At the beginning, \texttt{R2LTO} and \texttt{RLBR} seem to perform fairly well against \texttt{PDR}, but eventually, they start performing poorly as well. 

\begin{figure*}[!t]
\centering
     \begin{subfigure}[b]{0.32\textwidth} 
         \centering
         \includegraphics[width=\textwidth]{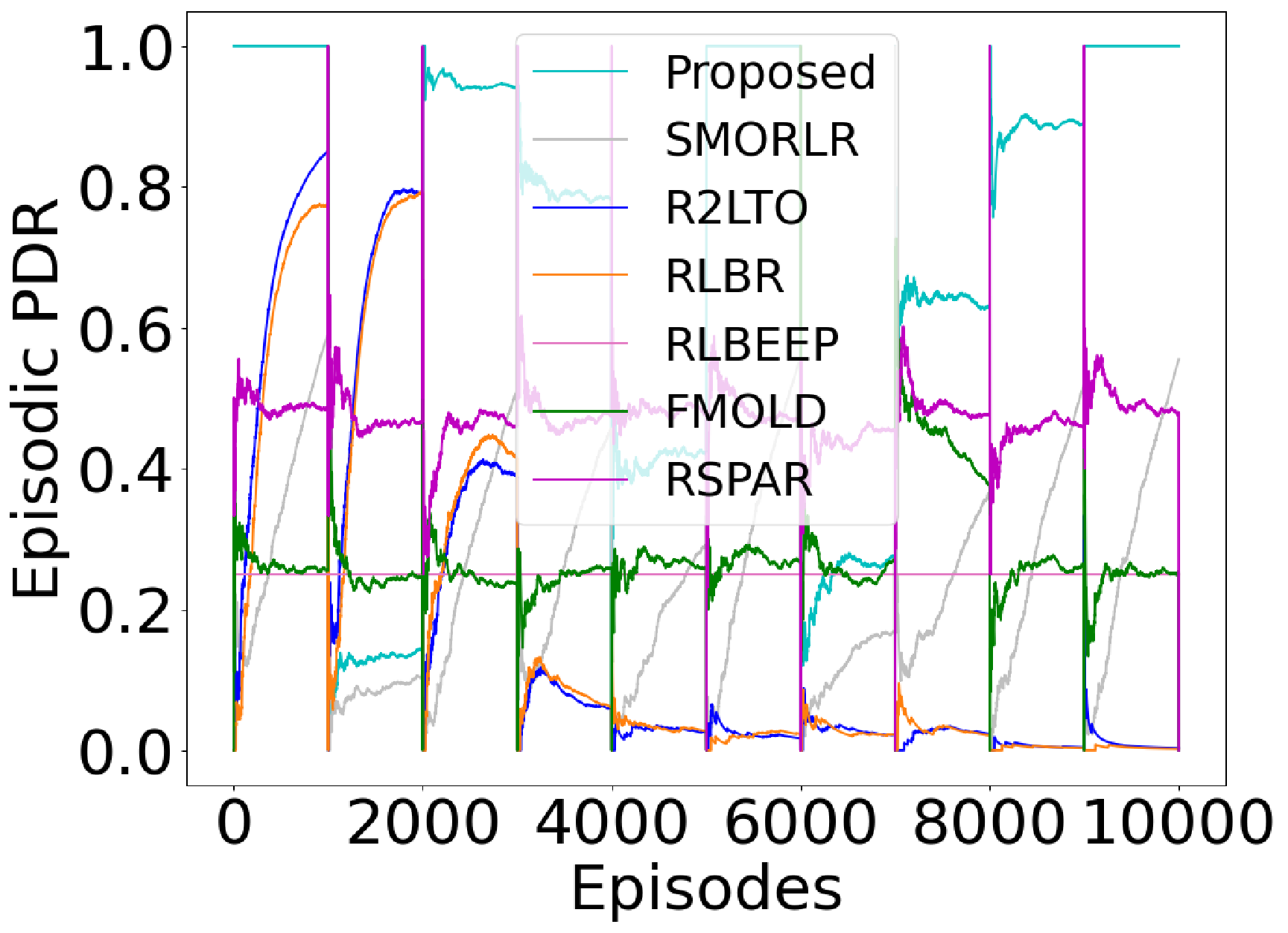}
         \caption{Experiment 1: Episodic PDR}
         \label{fig:delivery_dynamic_offline}
     \end{subfigure}
     \hfill
     \begin{subfigure}[b]{0.32\textwidth} 
         \centering
         \includegraphics[width=\textwidth]{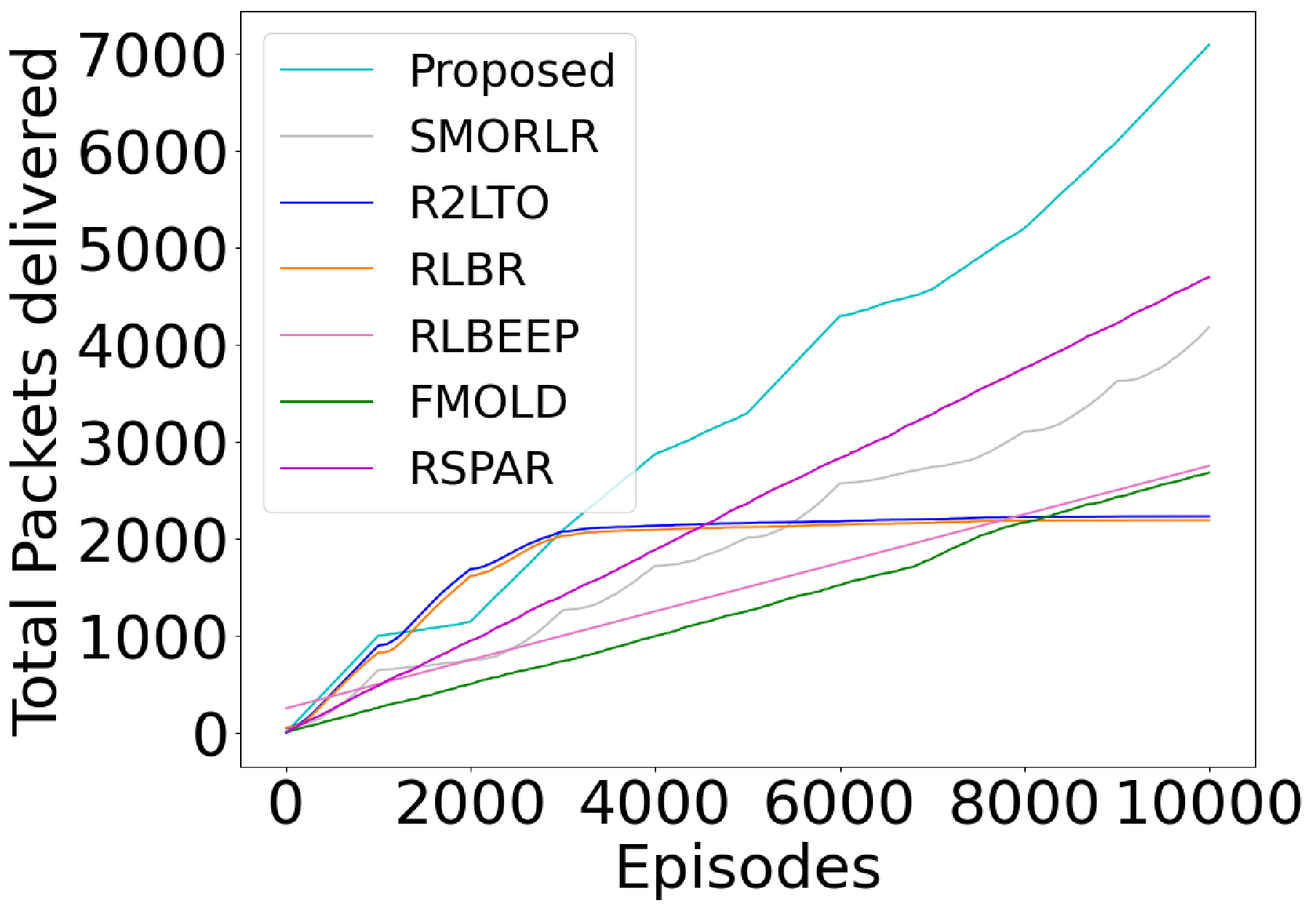}
         \caption{Exp. 1: Packets Delivered v/s Episodes}
         \label{fig:total_delivery_offline}
     \end{subfigure}
     \hfill
     \begin{subfigure}[b]{0.32\textwidth}
         \centering
         \includegraphics[width=\textwidth]{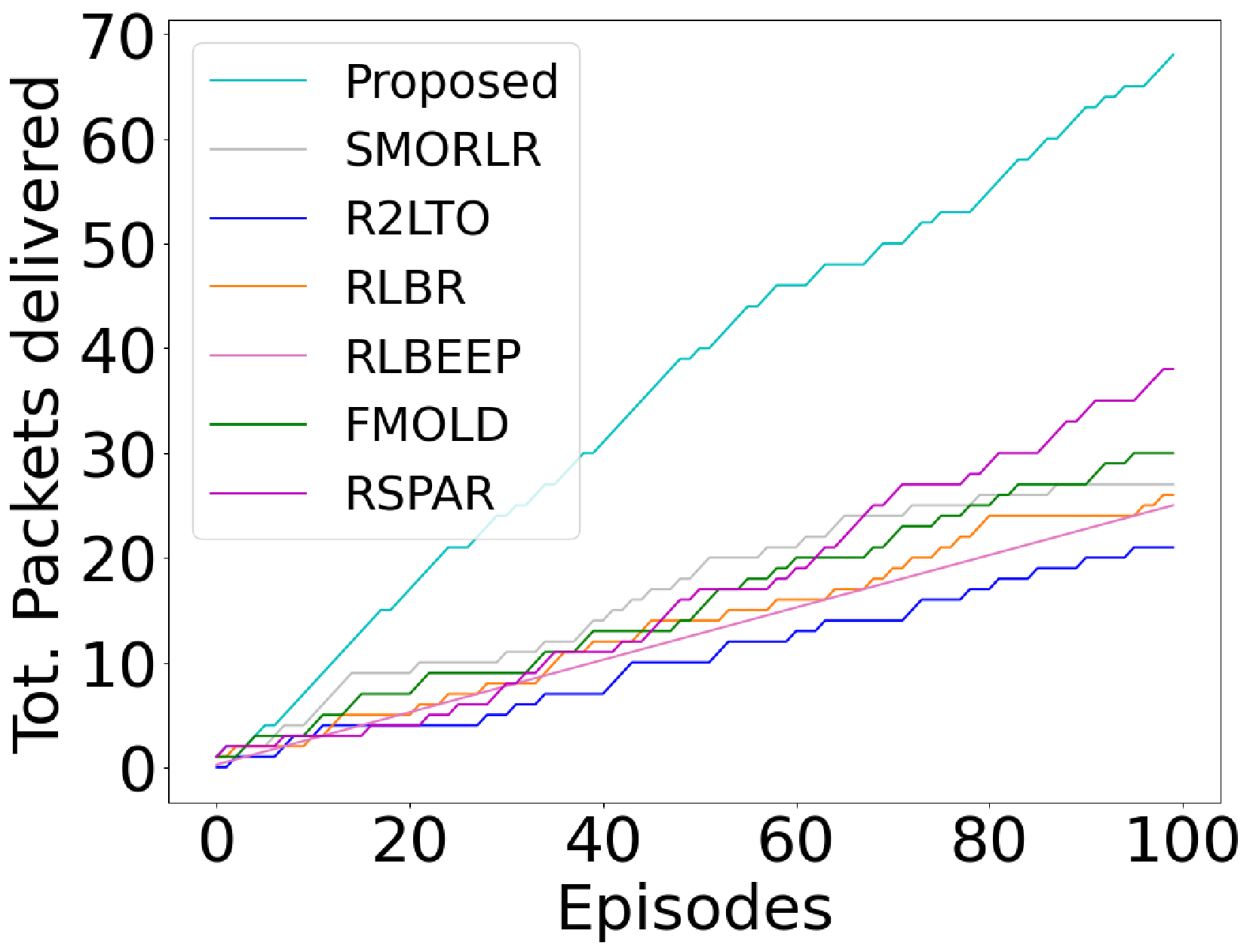}
         \caption{Exp. 2: Packets Delivered v/s Episodes}
         \label{fig:total_delivery_offline_cont_change}
     \end{subfigure}
   \caption{PDR in the Sequential exploration-exploitation scheme.}
        \label{fig:Results_Packet_Delivary_ratio}
\end{figure*}

We analyze cumulative packet delivery over a long run as shown in Figure \ref{fig:total_delivery_offline}, and the proposed method consistently outperforms all the baselines except for some initial fluctuations durations. As time passes, the difference between performances becomes more evident. At the end of $10000$ episodes, the proposed method has better performance than \texttt{RSPAR} by around $1.78\times$, \texttt{FMOLD} and \texttt{RLBEEP} by $3\times$, and \texttt{R2LTO} and \texttt{RLBR} by around $3.5\times$. Similar results were observed in Experiment 2, where preferences were randomly changed in each episode as shown in Figure \ref{fig:total_delivery_offline_cont_change}. We observe here that the total number of packets delivered by the proposed work consistently exceeds those delivered by the baseline method, and that the gap widens as more episodes are completed. By the end of $100$ episodes, the proposed algorithm's accumulated return reaches approximately $68$, which is $1.41\times$ that of \texttt{RSPAR}, $2.6\times$ greater than \texttt{FMOLD}, $2.62\times$ better than \texttt{R2LTO} and \texttt{RLBR}, $2.83\times$ higher than \texttt{SMORLR}, and $3.1\times$ higher than \texttt{RLBEEP}. Consequently, regardless of the exploration policy or preference variation patterns, the proposed method outperforms the baseline algorithms in terms of \texttt{PDR}.

\subsection{Simultaneous Exploration and Exploitation}\label{dynamic-online}

We further extended our experiments using a Simultaneous exploration and exploitation approach with varying frequency of preference variation at every thousandth episode (in Experiment $3$) and every episode (in Experiment $4$).
Similar to previous subsection, we evaluate the performance of the proposed Distributed DPQ-Learning routing algorithm against the six baseline routing algorithms based on three key criteria: overall reward, packet delivery, and energy consumption.

\subsubsection{Overall Reward}\label{dynamic-online-reward}
\begin{figure*}[t]
     \begin{subfigure}[b]{0.32\textwidth} 
         \centering
         \includegraphics[width=\textwidth]{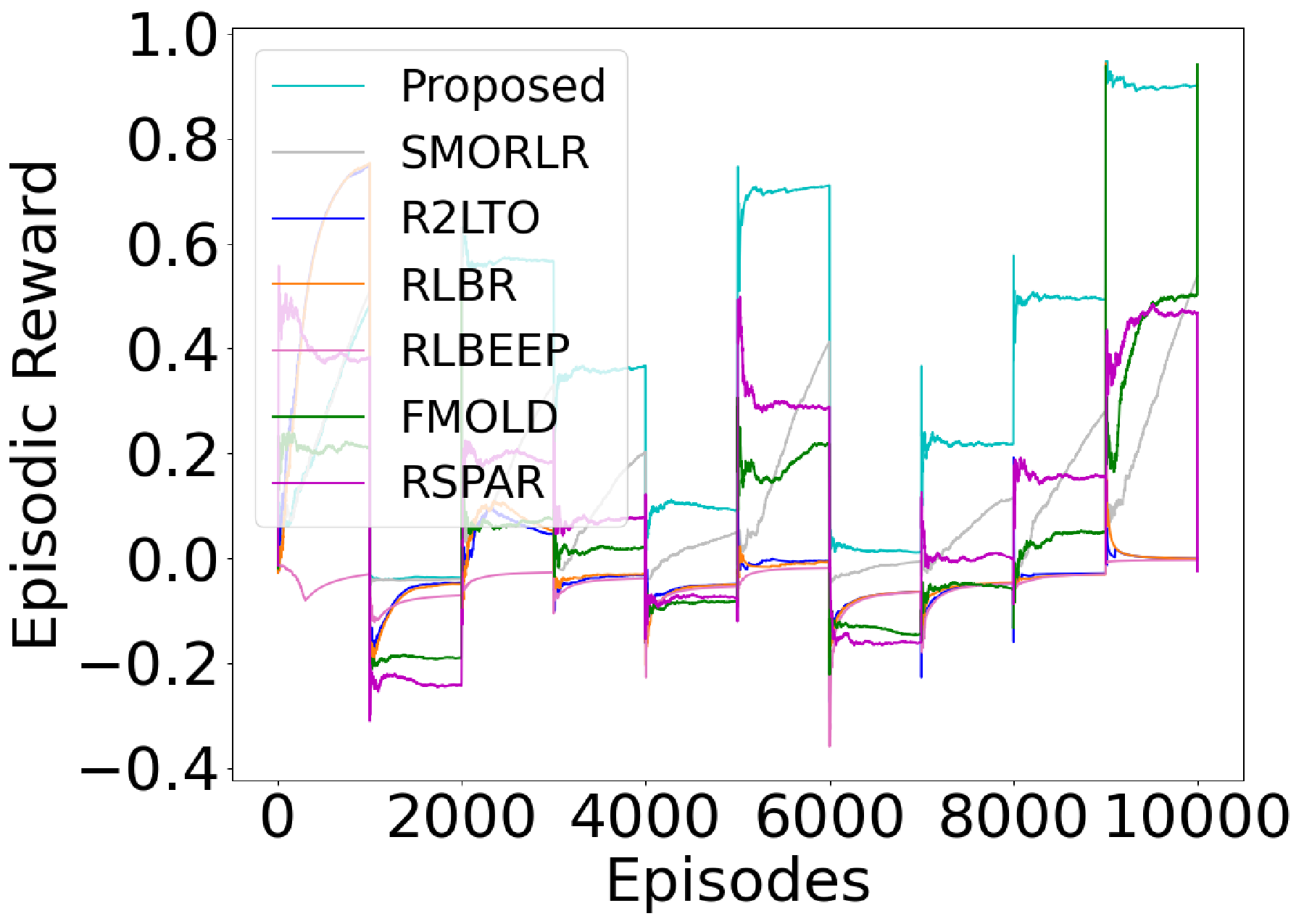}
\caption{Experiment 3: Comparison of the episodic rewards}\label{fig:reward_dynamic_online}
     \end{subfigure}
     \hfill
     \begin{subfigure}[b]{0.32\textwidth} 
         \centering
         \includegraphics[width=\textwidth]{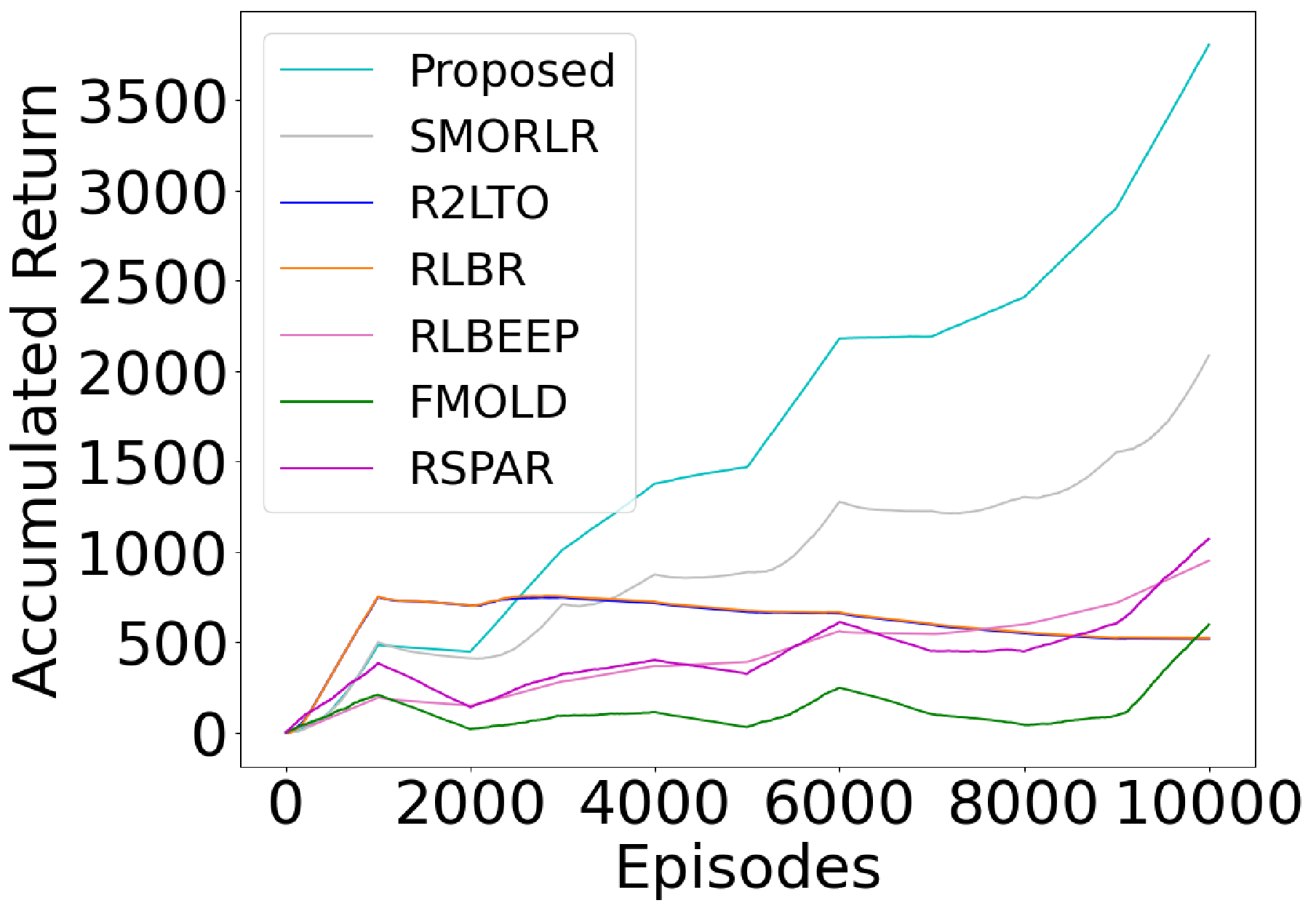}
\caption{Experiment 3: Comparison of accumulated returns}\label{fig:acc_return_dynamic_online}
     \end{subfigure}
        \hfill
  \begin{subfigure}[b]{0.32\textwidth}
         \centering
         \includegraphics[width=\textwidth]{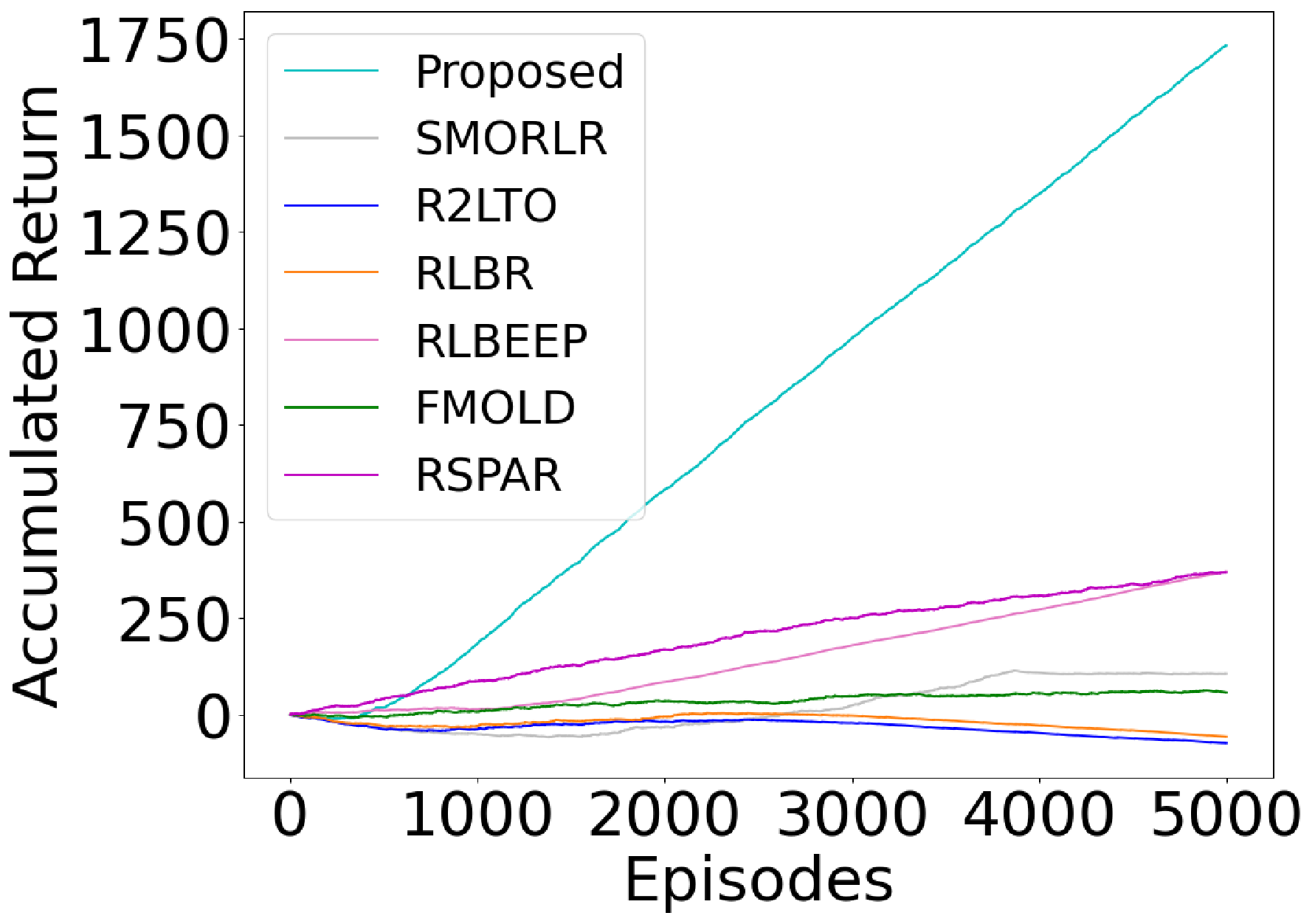}
         \caption{Experiment 4: Accumulated Return v/s Episodes}\label{fig:accumulated_Return_online_cont_change}
     \end{subfigure}    
   \caption{Overall Reward in the Simultaneous exploration-exploitation scheme.}
        \label{fig:Results_dynamic_online_Simultanious_OR}
\end{figure*}
We present both episodic rewards (Figure~\ref{fig:reward_dynamic_online}) and accumulated returns (Figure~\ref{fig:acc_return_dynamic_online}) of the proposed Distributed DPQ-Learning routing algorithm against the six baseline routing algorithms. From Figure~\ref{fig:reward_dynamic_online}, we observe that the proposed algorithm learns to optimize the overall return by performing more exploration than exploitation in the first $1000$ episodes, during which the Pareto front is also learned. In contrast, the simultaneous exploration--exploitation scheme carries out both processes in parallel, while gradually reducing exploration and increasing exploitation. 

From Figure~\ref{fig:reward_dynamic_online}, we note that the proposed algorithm may not be superior in the initial exploration phase, but after roughly $1000$ episodes it outperforms all baselines in episodic rewards. The proposed protocol also adjusts rapidly to changing preferences after limited exploration, while the baselines fail to do so. For example, \texttt{SMORLR} adapts after each preference change, but the adaptation is slow and preferences change again before it reaches the proposed protocol's performance. Further, \texttt{R2LTO}, \texttt{RLBEEP}, \texttt{FMOLD}, and \texttt{RLBR} consistently focus on a fixed objective preference, largely prioritizing energy efficiency. They learn certain intervals that improve returns, but still fall short of the proposed method by approximately $12$--$13\times$. \texttt{RSPAR} performs, on average, $344\%$ less efficiently than the proposed approach. Thus, while some baselines may perform well under specific preferences, the proposed algorithm consistently outperforms them and adapts quickly to changing preferences.

Figure \ref{fig:acc_return_dynamic_online} illustrates cumulative performance over time. Initially, all (baseline and proposed) algorithms start with zero cumulative returns. At the beginning, \texttt{RSPAR} and \texttt{RLBR} outperform the proposed DPQ-Learning algorithm. However, after a few explorations, the proposed approach starts utilizing the learned Pareto front and consistently outperforms the baselines, with the performance gap widening as more episodes are completed. Figure \ref{fig:acc_return_dynamic_online} shows that after $10,000$ episodes, the accumulated return of the proposed method surpasses \texttt{SMORLR} by $111.11\%$, \texttt{RSPAR} by $322.22\%$, \texttt{R2LTO} and \texttt{RLBEEP} by $375\%$, \texttt{FMOLD} by $850\%$, and \texttt{RLBR} by $927\%$ approximately. Similarly, significant improvements were seen in Experiment $4$, where preferences were randomly altered in each episode. Figure \ref{fig:accumulated_Return_online_cont_change} shows that the proposed routing protocol consistently achieves significantly higher accumulated returns than the baseline, with the performance gap growing as episodes progress. By the end of $5000$ episodes, the proposed algorithm achieves an accumulated return of around $1750$, which is $5.48\times$ that of \texttt{RSPAR} and \texttt{RLBEEP},  $8.21\times$ that of \texttt{SMORLR}, $34\times$ that of \texttt{FMOLD}, and $36\times$ that of \texttt{R2LTO} and \texttt{RLBR}. Based on these results, we observe that the proposed method consistently outperforms all baselines in terms of overall returns and adapts rapidly to changes in preferences regardless of exploration policy.

\subsubsection{Energy}\label{dynamic-online-energy} 
\begin{figure*}[t]
  \begin{subfigure}[b]{0.32\textwidth} 
         \centering
         \includegraphics[width=\textwidth]{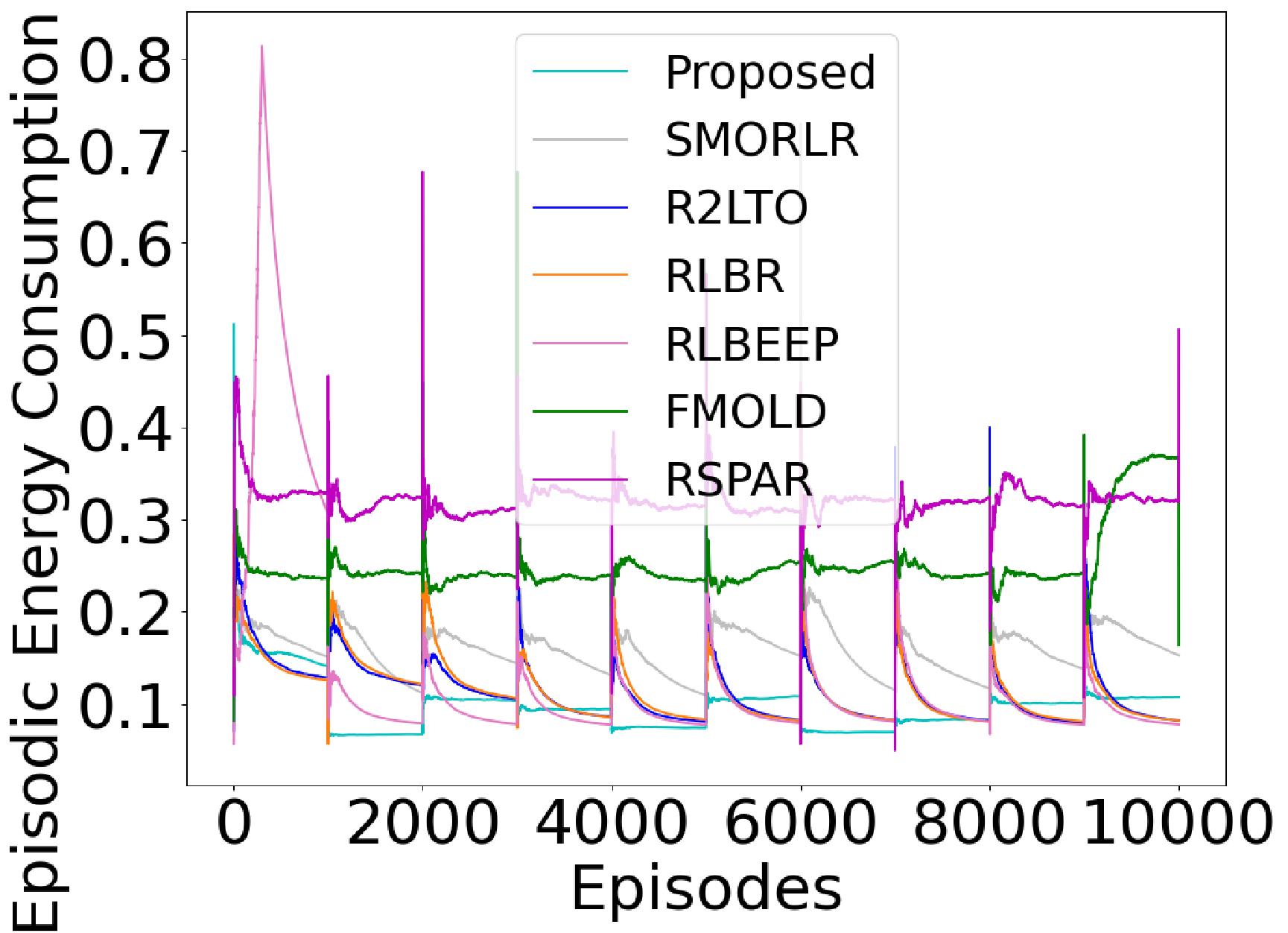}
         \caption{Experiment 3: Episodic Energy Consumption}
         \label{fig:energy_dynamic_online}
     \end{subfigure}
      \hfill
      \begin{subfigure}[b]{0.32\textwidth} 
         \centering
         \includegraphics[width=\textwidth]{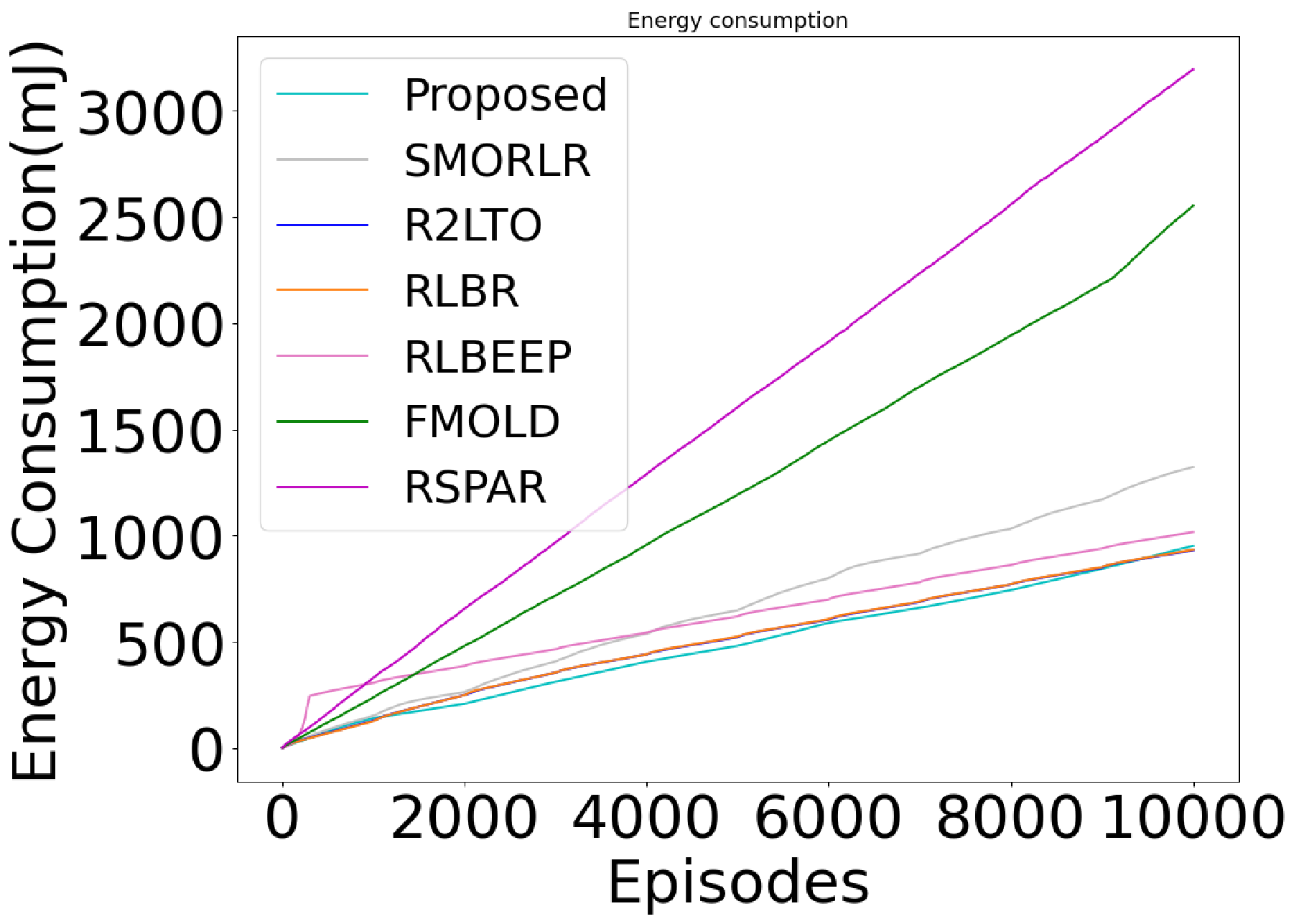}
\caption{Experiment 3: Comparison of total energy consumption.} \label{fig:online_total_Energy.eps}
     \end{subfigure}
      \hfill
    \begin{subfigure}[b]{0.32\textwidth} 
         \centering
         \includegraphics[width=\textwidth]{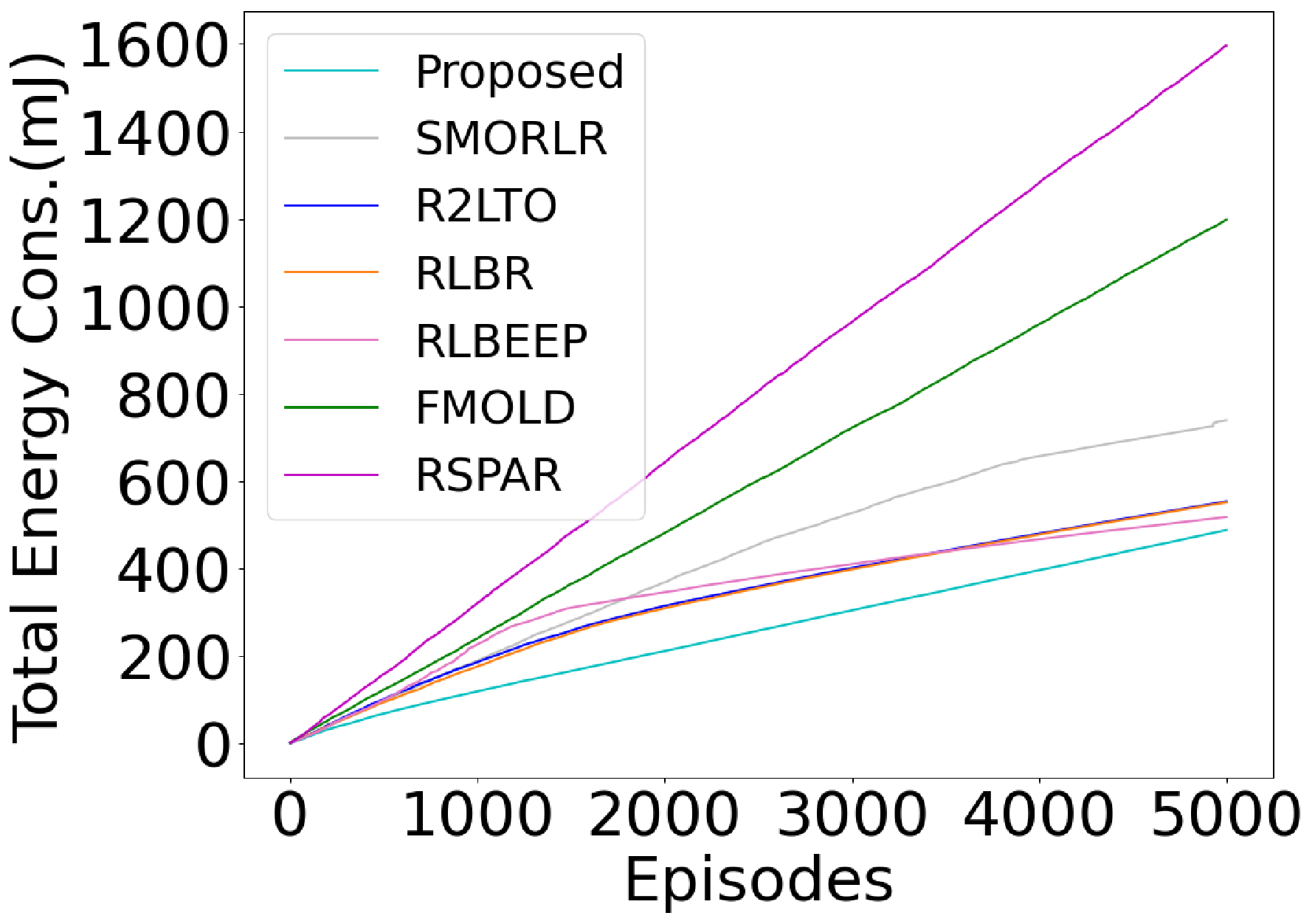}
         \caption{Experiment 4: Total Energy Consump. v/s Episodes}
         \label{fig:energy_online_cont_change}
     \end{subfigure}
   \caption{Energy Consumption in the Simultaneous exploration-exploitation scheme.}
        \label{fig:Results_dynamic_online_EC}
\end{figure*}
In this subsection, we assess the energy-efficiency of the proposed and baseline routing algorithms on individual objectives when preferences between them change periodically and frequently. Figure \ref{fig:energy_dynamic_online} highlights the episodic energy consumption (in mJ) across various routing algorithms over $10,000$ episodes, where the preference changes after every $1000$ episode. On average, the proposed distributed DPQ-Learning routing protocol demonstrates $65-80$\% lower energy consumption than the other routing protocols. The proposed approach consumed the minimum energy whenever a higher preference was given to the energy objective (intervals $1000-2000, 4000-5000$, and $6000-8000$) while performing reasonably well even when the preference for the energy objective was lower. On average, \texttt{SMORLR} consumes approximately $80\%$ more energy than the proposed method, although it gradually converges closer to it with each preference shift. However, the more frequently preferences change, the less efficient \texttt{SMORLR} becomes. \texttt{R2LTO}, \texttt{RLBEEP}, and \texttt{RLBR} consume around $60\%$ more energy than the proposed method. Since they focus primarily on energy considerations, they perform the highest among baselines. However, the proposed approach outperforms them because of a more efficient reward formulation. \texttt{FMOLD} consumes around $160\%$ more energy than the proposed method while maintaining a consistent performance. \texttt{RSPAR} exhibits the energy inefficiency $\approx200\%$ more energy consumption than the proposed algorithm.

The cumulative energy consumption shown in Figure \ref{fig:online_total_Energy.eps} indicate that after 10000 episodes, the cumulative energy consumption of the proposed method is very close to \texttt{RLBEEP}, \texttt{R2LTO}, and \texttt{RLBR}  with the value of around $700$ mJ, which is  $48.5\%$ lesser than that of \texttt{SMORLR} and $242.85\%$ less than that of \texttt{FMOLD}, and $242.85\%$ less than that of \texttt{RSPAR}. While \texttt{SMORLR}, \texttt{RLBEEP}, \texttt{R2LTO}, and \texttt{RLBR} perform reasonably well because of energy considerations in their RL reward formulations. Whereas \texttt{RSPAR} and \texttt{FMOLD} exhibit the highest cumulative energy consumption, underscoring their inefficiency. The plots demonstrate that even when energy is assigned to a lower preference, the overall energy consumption of the proposed routing protocol is still relatively low or equal to most baselines. A similar trend was observed in Experiment 4, where preferences were randomly altered in each episode. As seen in Figure \ref{fig:energy_online_cont_change}, the proposed routing protocol's cumulative energy consumption remains more efficient than that of the baseline methods, with the gap increasing as the episodes progress. After $5000$ episodes, the cumulative energy consumption of the proposed distributed DPQ-Learning routing algorithm is $78.125\%$ lower than that of \texttt{RSPAR}, $66.35\%$ lower than that of \texttt{FMOLD}, $41.67\%$ lower than that of \texttt{SMORLR}, $12.5\%$ lower than that of \texttt{R2LTO}, and \texttt{RLBR}, and $3\%$ lower than that of \texttt{RLBEEP}. Therefore, the proposed method consumes significantly less energy than baseline algorithms regardless of exploration strategy or preference shifts. Whenever an energy objective is given higher priority, the method consumes the least amount of energy and performs reasonably well even when the energy objective is given a lower priority.

\subsubsection{Packet Delivery}\label{dynamic-online-pdr}
Here, we assess PDR of the proposed and baseline algorithms on individual objectives when preferences between them change periodically (every 1000 episodes) and frequently (each episode). Figures \ref{fig:delivery_dynamic_online} and \ref{fig:online_total_packets_delivered.eps} illustrate variations in episodic and cumulative packet deliveries, respectively. From Figure \ref{fig:delivery_dynamic_online}, it is evident that after a certain duration of exploration ($1000$ episodes), the proposed distributed DPQ-Learning routing algorithm outperforms all baselines in episodes where the preference for the \texttt{PDR} objective exceeds 0.5. 
In comparison, \texttt{RSPAR} and \texttt{FMOLD} maintain episodic \texttt{PDR} levels of around 0.4-0.5 and 0.2-0.3, respectively. Although \texttt{SMORLR} gradually learns to maximize \texttt{PDR} with each new preference vector, its peak performance remains 50\% lower than the consistent results achieved by the proposed algorithm (after the exploration period). Meanwhile, \texttt{RLBEEP}, \texttt{R2LTO}, and \texttt{RLBR} primarily focus on energy optimization and disregard the \texttt{PDR} objective, which consistently underperforms. Initially, \texttt{R2LTO} and \texttt{RLBR} demonstrate reasonable \texttt{PDR} performance, but it deteriorated over time.

\begin{figure*}[t]
     \begin{subfigure}[b]{0.32\textwidth} 
         \centering
         \includegraphics[width=\textwidth]{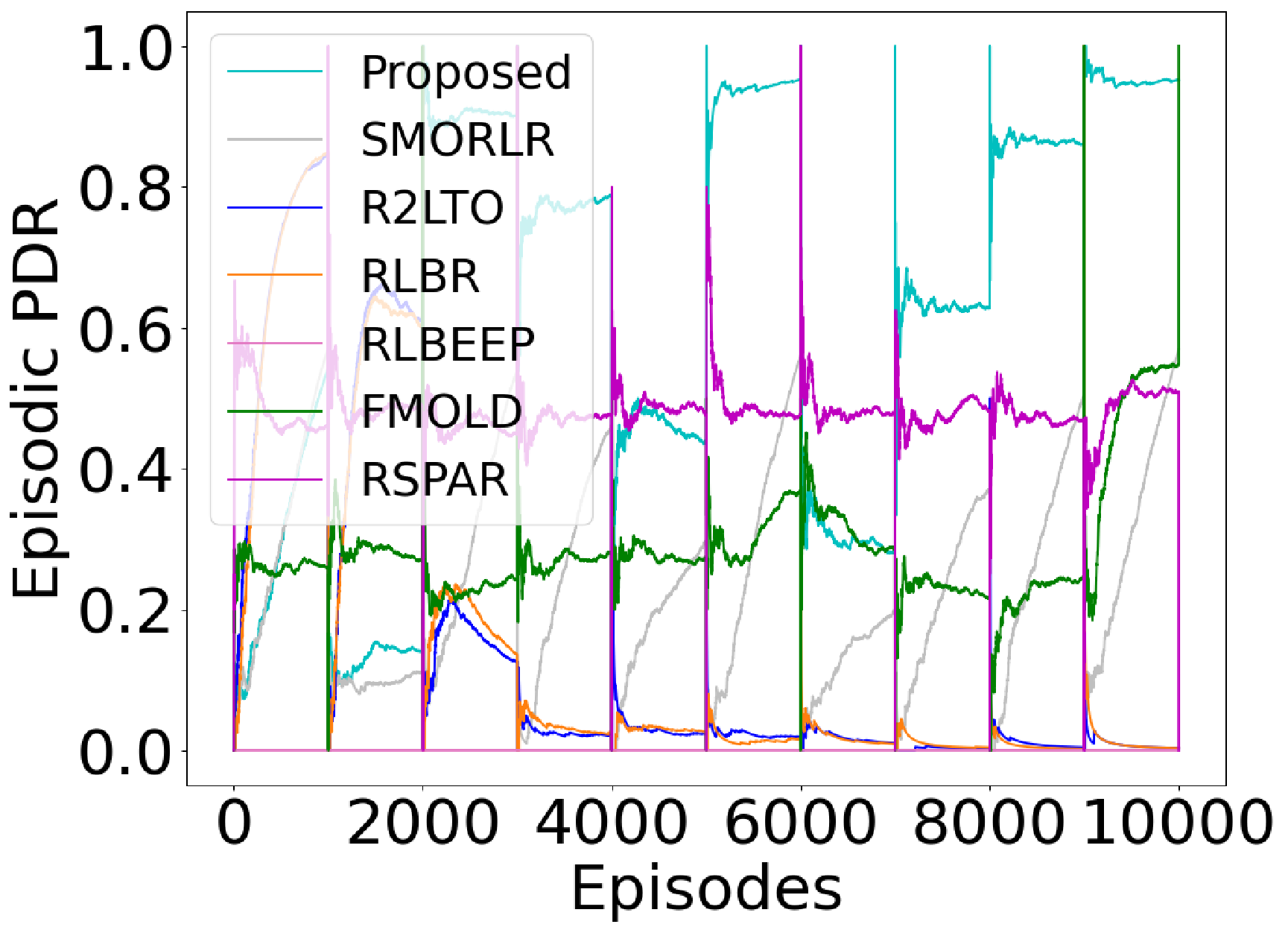}
         \caption{Experiment 3: Episodic PDR}
         \label{fig:delivery_dynamic_online}
     \end{subfigure}
           \hfill
     \begin{subfigure}[b]{0.32\textwidth} 
         \centering
         \includegraphics[width=\textwidth]{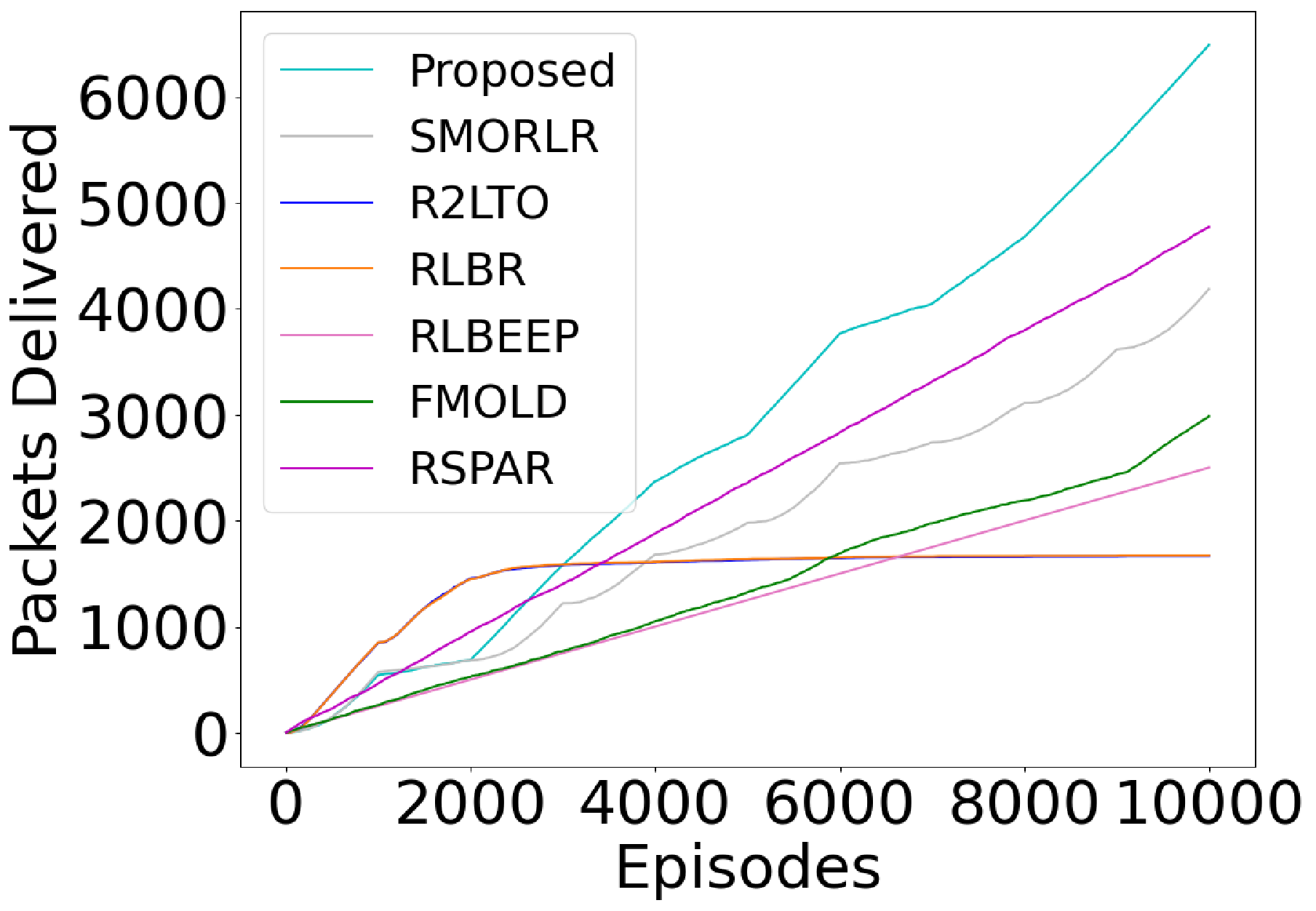}
         \caption{Experiment 3: Comparison of total packets delivered}
         \label{fig:online_total_packets_delivered.eps}
     \end{subfigure}
     \hfill
     \begin{subfigure}[b]{0.32\textwidth}
         \centering
         \includegraphics[width=\textwidth]{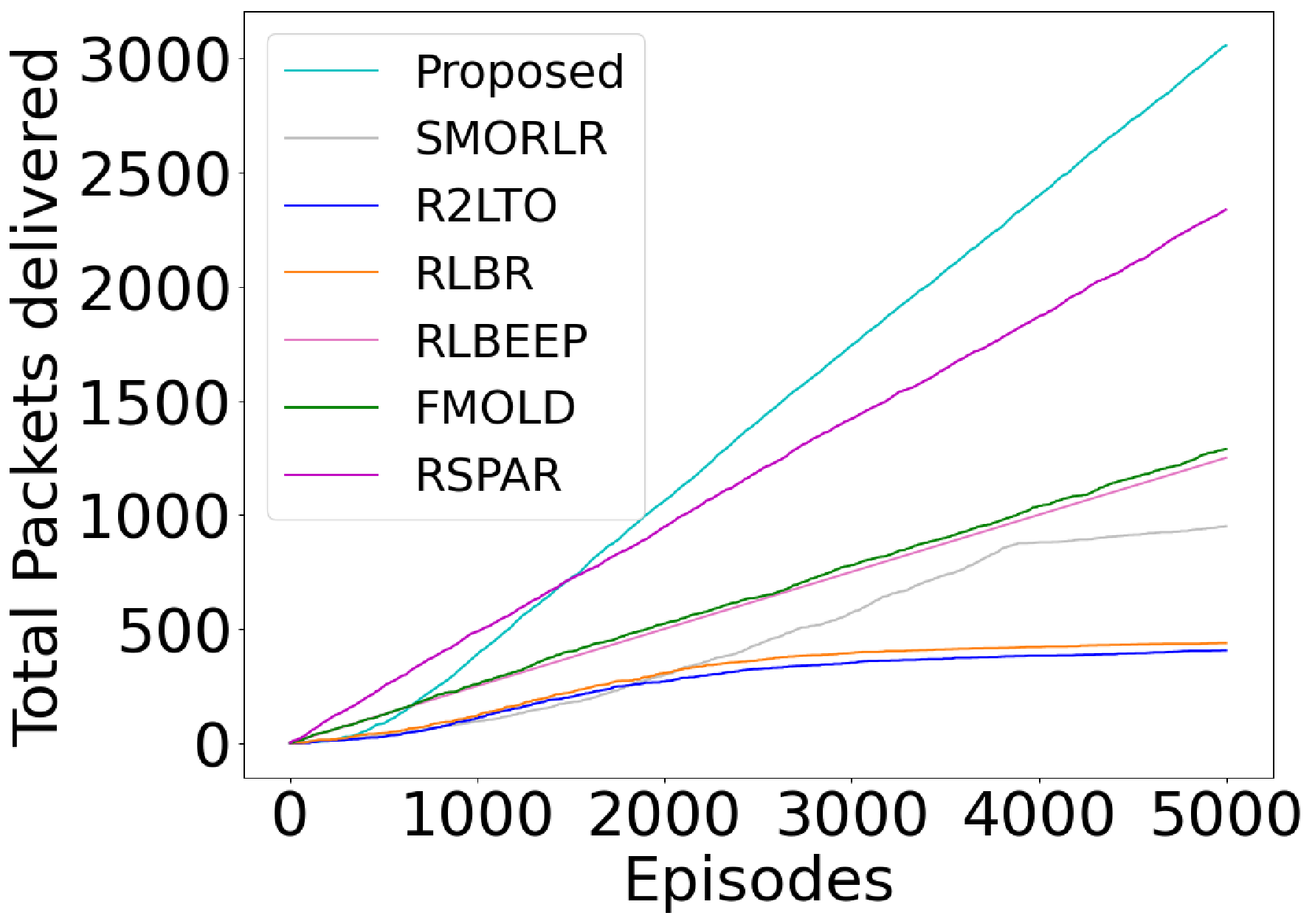}
         \caption{Experiment 4: Total Packets Delivered v/s Episodes}
         \label{fig:delivery_online_cont_change}
     \end{subfigure}
   \caption{Packet Delivery Ratio in the Simultaneous exploration-exploitation scheme.}
        \label{fig:Results_dynamic_online_PDR}
\end{figure*}

Figure~\ref{fig:online_total_packets_delivered.eps} provides a more comprehensive view through cumulative packet delivery, enabling a fair long-run comparison. The proposed method begins to outperform all baselines after $\approx 3000$ episodes of sufficient learning, and the gap widens over time. By the end of $10{,}000$ episodes, it delivers $46.67\%$ more packets than \texttt{RSPAR}, $69.23\%$ more than \texttt{SMORLR}, $164\%$ more than \texttt{FMOLD}, $230\%$ more than \texttt{RLBR}, and $255.12\%$ more than \texttt{R2LTO} and \texttt{RLBEEP}. Similar trends appear in Experiment~4 with per-episode random preference changes. As shown in Figure~\ref{fig:delivery_online_cont_change}, the proposed protocol consistently exceeds all baselines, with the separation increasing as episodes progress. At $5000$ episodes, it reaches around $3200$ delivered packets, which is $33.33\%$ higher than \texttt{RSPAR}, around $166\%$ higher than \texttt{FMOLD} and \texttt{RLBR}, $220\%$ higher than \texttt{SMORLR}, and $814.3\%$ higher than \texttt{RLBEEP} and \texttt{R2LTO}. Overall, the proposed method performs best when the preference for \texttt{PDR} is high, while remaining acceptable for lower preferences.

\subsection{Sensitivity Analysis}
\label{subsec:sensitivity}

To assess the robustness of the proposed Dynamic Preference-based Q-learning (DPQ) routing framework, we conducted a sensitivity analysis by varying the temporal aggregation window used during testing. This analysis evaluates how performance changes when network preferences persist for shorter or longer durations, which is representative of heterogeneous and time-varying IoT operating conditions.

\paragraph{Experimental Setup}
All experiments were conducted in the same wireless sensor network environment described earlier. After a fixed training phase, routing policies were evaluated over testing windows of length $W \in \{50, 200, 500\}$ episodes. To evaluate sensitivity with respect to preference discretization, the proposed method was tested using two different preference grids: a \emph{coarse} grid with two anchor points, $\beta \in \{0, 1\}$, and a \emph{fine} grid with eleven anchor points, $\beta \in \{0, 0.1, \ldots, 1\}$. All other parameters, including network topology, traffic patterns, and random seeds, were kept identical across methods to ensure a fair comparison.

The cumulative reward, which jointly captures packet delivery success and energy consumption, is the primary optimization objective of the proposed method. Accordingly, reward is reported first and highlighted in the tables, while packet delivery ratio (PDR) and energy consumption are included as supporting metrics.

\paragraph{Compared Methods}
We compare two variants of the proposed approach—DPQ with a coarse preference grid (DPQ coarse\_2) and a finer grid (DPQ fine\_11)—against representative reinforcement-learning-based baselines (R2LTO, RLBR, RLBEEP) and a static multi-objective Q-learning method (SMORLR). The coarse and fine grids correspond to two-point and eleven-point preference discretizations, respectively.

\paragraph{Results and Discussion}
Tables~\ref{tab:sens_w50}--\ref{tab:sens_w500} summarize the sensitivity results. Across all window lengths, the proposed DPQ variants consistently achieve the highest cumulative reward, demonstrating robustness to changes in preference persistence. While certain baselines occasionally achieve comparable PDR or lower energy consumption, they fail to balance these objectives effectively, resulting in significantly lower reward. This confirms that optimizing individual metrics in isolation does not yield optimal long-term routing performance.

\begin{table}[t]
\centering
\caption{Sensitivity results for window length $W=50$ (mean $\pm$ std).}
\label{tab:sens_w50}
\begin{tabular}{lccc}
\hline
Method & Reward ($\uparrow$) & PDR ($\uparrow$) & Energy ($\downarrow$) \\
\hline
DPQ coarse\_2 & \textbf{0.863 $\pm$ 0.236} & 0.954 $\pm$ 0.168 & 0.196 $\pm$ 0.331 \\
DPQ fine\_11 & \textbf{0.862 $\pm$ 0.236} & 0.953 $\pm$ 0.161 & 0.110 $\pm$ 0.011 \\
SMORLR (Static-Q) & 0.120 $\pm$ 0.066 & 0.141 $\pm$ 0.056 & 0.193 $\pm$ 0.026 \\
R2LTO & 0.120 $\pm$ 0.063 & 0.142 $\pm$ 0.052 & 0.182 $\pm$ 0.020 \\
RLBR & 0.120 $\pm$ 0.063 & 0.142 $\pm$ 0.052 & 0.182 $\pm$ 0.020 \\
RLBEEP & 0.120 $\pm$ 0.063 & 0.142 $\pm$ 0.052 & 0.182 $\pm$ 0.020 \\
\hline
\end{tabular}
\end{table}

\begin{table}[t]
\centering
\caption{Sensitivity results for window length $W=200$ (mean $\pm$ std).}
\label{tab:sens_w200}
\begin{tabular}{lccc}
\hline
Method & Reward ($\uparrow$) & PDR ($\uparrow$) & Energy ($\downarrow$) \\
\hline
DPQ coarse\_2 & \textbf{0.610 $\pm$ 0.374} & 0.815 $\pm$ 0.298 & 0.437 $\pm$ 0.579 \\
DPQ fine\_11 & \textbf{0.605 $\pm$ 0.374} & 0.809 $\pm$ 0.285 & 0.101 $\pm$ 0.017 \\
SMORLR (Static-Q) & 0.292 $\pm$ 0.223 & 0.403 $\pm$ 0.158 & 0.140 $\pm$ 0.025 \\
R2LTO & 0.350 $\pm$ 0.185 & 0.545 $\pm$ 0.031 & 0.150 $\pm$ 0.005 \\
RLBR & 0.350 $\pm$ 0.185 & 0.545 $\pm$ 0.031 & 0.150 $\pm$ 0.005 \\
RLBEEP & 0.350 $\pm$ 0.185 & 0.545 $\pm$ 0.031 & 0.150 $\pm$ 0.005 \\
\hline
\end{tabular}
\end{table}

\begin{table}[t]
\centering
\caption{Sensitivity results for window length $W=500$ (mean $\pm$ std).}
\label{tab:sens_w500}
\begin{tabular}{lccc}
\hline
Method & Reward ($\uparrow$) & PDR ($\uparrow$) & Energy ($\downarrow$) \\
\hline
DPQ coarse\_2 & \textbf{0.453 $\pm$ 0.346} & 0.737 $\pm$ 0.346 & 0.516 $\pm$ 0.653 \\
DPQ fine\_11 & \textbf{0.439 $\pm$ 0.346} & 0.713 $\pm$ 0.332 & 0.095 $\pm$ 0.019 \\
SMORLR (Static-Q) & 0.232 $\pm$ 0.227 & 0.417 $\pm$ 0.212 & 0.101 $\pm$ 0.013 \\
R2LTO & 0.413 $\pm$ 0.246 & 0.819 $\pm$ 0.011 & 0.129 $\pm$ 0.003 \\
RLBR & 0.413 $\pm$ 0.246 & 0.819 $\pm$ 0.011 & 0.129 $\pm$ 0.003 \\
RLBEEP & 0.413 $\pm$ 0.246 & 0.819 $\pm$ 0.011 & 0.129 $\pm$ 0.003 \\
\hline
\end{tabular}
\end{table}
\color{black}

\section{Conclusions} \label{sec:conclusion}
This paper presents a distributed multi-objective Q-learning–based routing algorithm for IoT and wireless sensor networks. The method learns multiple per-preference Q-tables in parallel and uses a greedy interpolation policy to adapt instantly to changing preference weights without retraining or central coordination. Theoretical analysis establishes a uniform near-optimality guarantee, explaining the stability and effectiveness of the proposed interpolation mechanism. Simulation results confirm consistent gains over existing schemes, demonstrating that the proposed Dynamic Preference Q-Learning framework provides a light-weight and theoretically grounded approach for real-time distributed IoT routing.The sensitivity analysis validates that the proposed method maintains superior reward performance across diverse settings, highlighting its robustness and suitability for dynamic IoT routing scenarios with time-varying objectives.

\bibliographystyle{IEEEtran}
\bibliography{references.bib}

\balance

\vskip -1\baselineskip plus -1fil
\begin{IEEEbiographynophoto}{Shubham Vaishnav}
is currently pursuing Ph.D. at the Department of Computer and Systems Science, Stockholm University, Sweden. He received Bachelors' and Masters' degrees in Computer Science \& Engineering from IIT (ISM), Dhanbad. His research includes Reinforcement learning, Federated Learning, and AI-driven Multiobjective decision-making with applications to IoT. 
\end{IEEEbiographynophoto}\vskip -1\baselineskip plus -1fil
\begin{IEEEbiographynophoto}{Praveen Kumar Donta} (SM'22), Associate Professor at the Department of Computer and Systems Sciences, Stockholm University, Sweden. He was a Postdoc at the Distributed Systems Group, TU Wien, Vienna, Austria. He received his Ph. D. at the IIT (ISM), Dhanbad in June 2021. His current research is on Distributed Computing Continuum Systems.
\end{IEEEbiographynophoto}\vskip -1\baselineskip plus -1fil
\begin{IEEEbiographynophoto}{Sindri Magnússon}
is an Associate Professor in the Department of Computer and Systems Science at Stockholm University, Sweden. He received a PhD in Electrical Engineering from KTH Royal Institute of Technology, Stockholm, Sweden, in 2017. He was a postdoctoral researcher 2018-2019 at Harvard University, Cambridge, MA. His research interests include large-scale distributed/parallel optimization, machine learning, and control.
\end{IEEEbiographynophoto}

\end{document}